\documentclass[letterpaper,11pt,reqno]{amsart}
\usepackage[utf8]{inputenc}

\makeatletter
\usepackage{silence}
\usepackage{amssymb}
\usepackage{latexsym}
\usepackage{amsbsy}
\usepackage{amsfonts}
\usepackage{amsthm}
\usepackage{hyperref}
\usepackage{graphicx}
\usepackage{enumerate}
\usepackage{enumitem}
\usepackage{mathtools}
\usepackage{color}
\usepackage{mathrsfs}
\usepackage{multirow}
\usepackage[authoryear,sort]{natbib}
\usepackage{listings}
\usepackage{xcolor} 
\usepackage{graphicx}
\usepackage{subcaption}
\usepackage{float}
\usepackage{placeins}

\allowdisplaybreaks[2]
\usepackage[top=1in, bottom=1in, left=1in, right=1in]{geometry}
\usepackage{setspace}
\usepackage{tikz}

\usepackage{fontawesome5}
\usetikzlibrary{arrows.meta, positioning, shapes.multipart, calc, matrix,fit,backgrounds}

\def\marginpar#1{\ignorespaces}

\newtheorem{theorem}{Theorem}[section]
\newtheorem{lemma}[theorem]{Lemma}
\newtheorem{proposition}[theorem]{Proposition}

\newtheorem{remark}[theorem]{Remark}
\newtheorem{assumption}[theorem]{Assumption}

\numberwithin{equation}{section}
\makeatother
\begin{document}
\title[Optimal Decisions for Liquid Staking: Allocation and Exit Timing]{Optimal Decisions for Liquid Staking: \\ Allocation and Exit Timing}

\author[Ruofei Ma, Zhebiao Cai, Wenpin Tang, David Yao]{{Ruofei} Ma, {Zhebiao} Cai, {Wenpin} Tang, {David} Yao}
\address{Department of Industrial Engineering and Operations Research, Columbia University. 
} 
\email{rm3881@columbia.edu, zc2801@columbia.edu, wt2319@columbia.edu, ddy1@columbia.edu}

\date{} 

\begin{abstract}
In this paper, we study an investor's optimal entry and exit decisions in a liquid staking protocol (LSP) and an automated market maker (AMM), primarily from the standpoint of the investor. Our analysis focuses on two key investor actions: the initial allocation decision at time $t=0$, and the optimal timing of exit. First, we derive an optimal allocation strategy that enables the investor to distribute risk across the LSP, AMM, and direct holding. Our results also offer insights for LSP and AMM designers, identifying the necessary and sufficient conditions under which the investor is incentivized to stake through an LSP, and further, to provide liquidity in addition to staking. These conditions include a lower bound on the transaction fee, for which we propose a fee mechanism that attains the bound. Second, given a fixed protocol design, we model the optimal exit timing of an individual investor using Laplace transforms and free-boundary techniques. We analyze scenarios with and without transaction fees. In the absence of fees, we decompose the investor's payoff into impermanent loss and opportunity cost, and provide theoretical results characterizing the investor's payoff and the optimal exit threshold. With transaction fees, we conduct numerical analyses to examine how fee accumulation influences exit strategies. Our results reveal that in both settings, a stop-loss strategy often maximizes the investor's expected payoff, driven by opportunity gains and the accumulation of fees where fees are present. Our analyses rely on various tools from stochastic processes and control theory, as well as convex optimization and analysis. We further support our theoretical insights with numerical experiments and explore additional properties of the investor's value function and optimal behavior.
\end{abstract}

\maketitle

\section{Introduction}
\quad   The rapid growth of decentralized finance (DeFi) and the transition of Ethereum to Proof-of-Stake (PoS) have spurred new financial innovations. A prominent example is the rise of liquid staking protocols (LSPs) – protocols that allow users to stake their ETH and receive derivative tokens, which are usually referred to as liquid staking tokens (LSTs), that remain liquid and usable in other applications \citep{xiong2025}. These LSTs represent claims on staked ETH (including accumulated rewards) and can be freely traded or deployed in DeFi \citep{scharnowski2025, xiong2025}. Notably, liquid staking builds upon Ethereum’s Proof-of-Stake (PoS) design, where validators are selected to propose and validate blocks based on the amount of ETH they stake. This contrasts with Bitcoin’s Proof-of-Work (PoW) system, in which block validation rights are determined by computational power. By enabling staking and generating yield-bearing tokens, liquid staking protocols offer a mechanism for users to participate in network security while maintaining liquidity. Common use cases of LSTs include lending on protocols like Aave \citep{heimbach2023defilendingmerge}, restaking in actively validated services such as EigenLayer \citep{team2024eigenlayer}, or providing liquidity in automated market makers (AMMs) \citep{gogol2024, xiong2024}. This paper focuses on the third use case of providing liquidity in AMMs using LSTs. Figure \ref{fig: 1} illustrates the various ways in which investors can utilize Liquid Staking Tokens (LSTs) across decentralized finance (DeFi) applications.



\begin{figure}[h] 
\centering
\begin{tikzpicture}[
  node distance=2cm and 1.6cm,
  box/.style={minimum width=1.5cm, minimum height=1cm, align=center},
  >=Stealth
]

\node[inner sep=5pt, xshift=-0.5em] (lsp)  {\includegraphics[width=1.2cm]{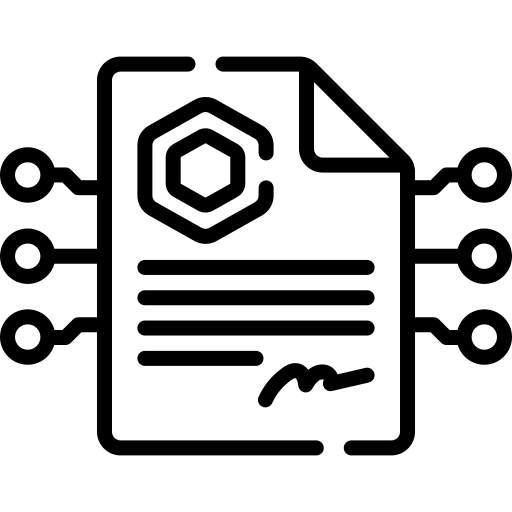}};

\node[align=center] at (lsp.south) [yshift=-1.5em] {\shortstack{Liquid Staking\\Protocol}};

\node[inner sep=5pt,xshift = 0.5em] (investor) [right=of lsp] {\includegraphics[width=1.2cm]{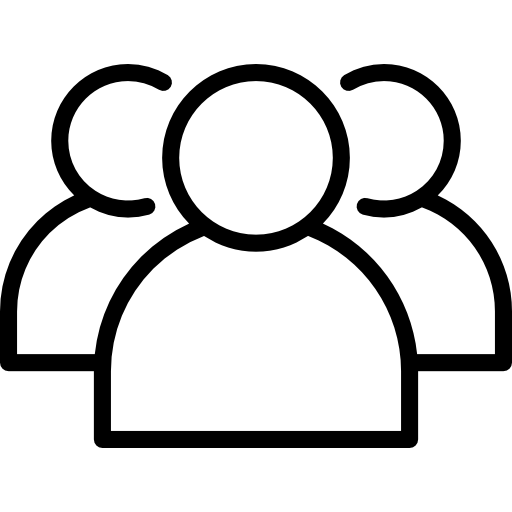}};

\node[align=center] at (investor.south) [yshift=-0.9em] {Investors};

\node[box, right=of investor, yshift=1.8cm, xshift = 1em] (lending) {
  \begin{tikzpicture}[baseline=(text.base)]
    \node (icon) at (0,0) {\includegraphics[height=1cm]{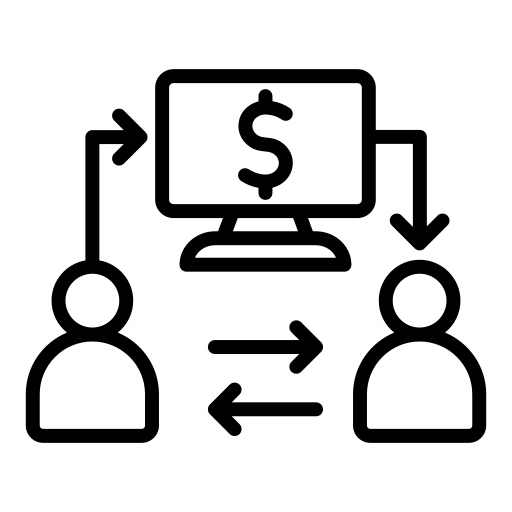}};
    \node[right=0.5em of icon, align=left] (text) {{\footnotesize Lending and Borrowing}};
  \end{tikzpicture}
};

\node[box, right=of investor, yshift=0.6cm, xshift = 1em] (liquidity) {
  \begin{tikzpicture}[baseline=(text.base)]
    \node (icon) at (0,0) {\includegraphics[height=1cm]{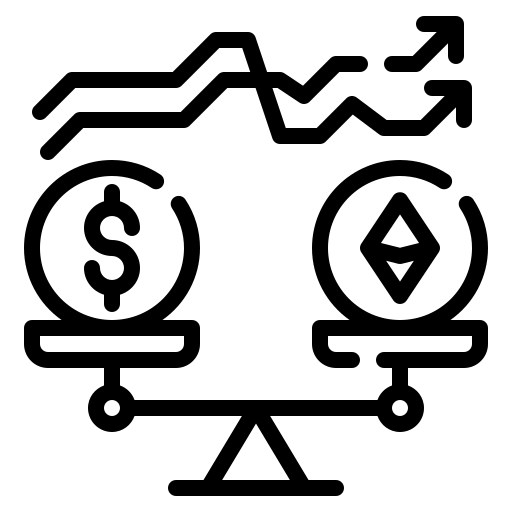}};
    \node[right=0.5em of icon, align=left] (text) {\textbf{Liquidity Provision}};
  \end{tikzpicture}
};

\node[box, right=of investor, yshift=-0.6cm, xshift = 1em] (trading) {
  \begin{tikzpicture}[baseline=(text.base)]
    \node (icon) at (0,0) {\includegraphics[height=1cm]{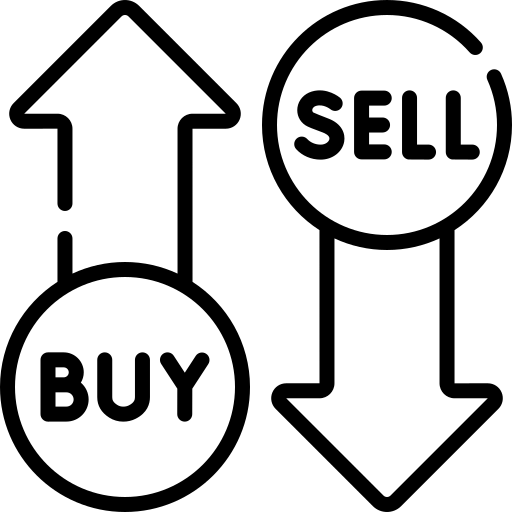}};
    \node[right=0.5em of icon, align=left] (text) {{\footnotesize Trading}};
  \end{tikzpicture}
};

\node[box, right=of investor, yshift=-1.8cm, xshift = 1em] (restaking) {
  \begin{tikzpicture}[baseline=(text.base)]
    \node (icon) at (0,0) {\includegraphics[height=1cm]{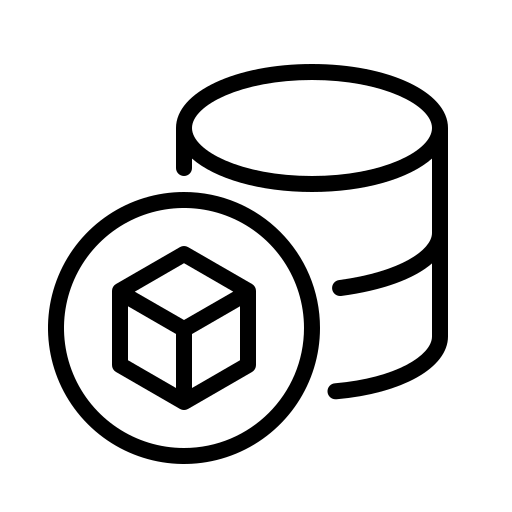}};
    \node[right=0.5em of icon, align=left] (text) {{\footnotesize Restaking}};
  \end{tikzpicture}
};

\begin{scope}[on background layer]
  \node[
    draw,
    dashed,
    rounded corners,
    inner sep=-3.5pt,
    fit=(lending)(liquidity)(trading)(restaking),
    label=above:{\textbf{Examples of LST Use Cases}}
  ] (servicebox) {};
\end{scope}

\draw[->] ([yshift=0.5em]lsp.east) -- ([yshift=0.5em]investor.west)
  node[midway, yshift=1em, xshift=-1em] {\includegraphics[width=0.5cm]{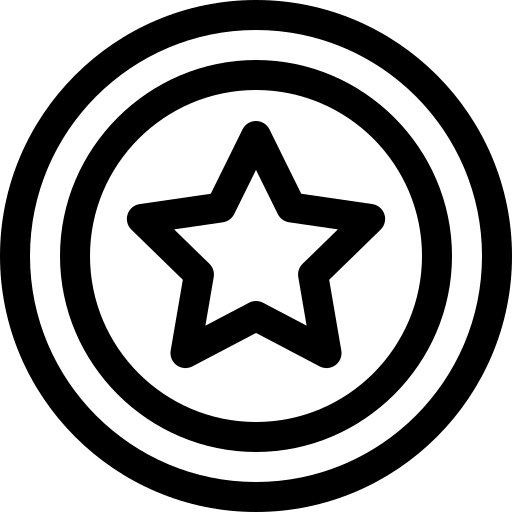}} 
  node[midway, yshift=1em,xshift = 1em] {LST};

\draw[->] ([yshift=-0.5em]investor.west) -- ([yshift=-0.5em]lsp.east)
  node[midway, yshift=-1em, xshift=-1em] {\includegraphics[width=0.5cm]{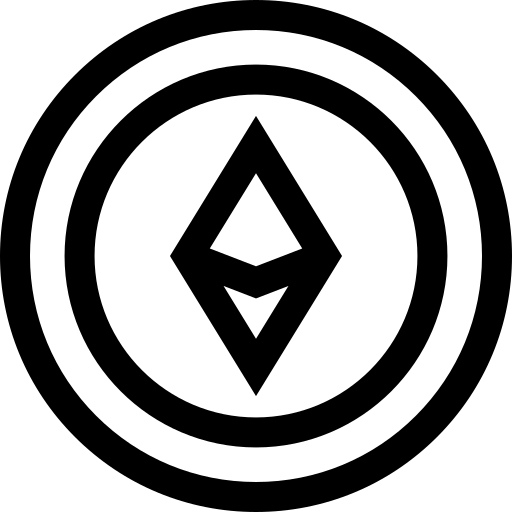}} 
  node[midway, yshift=-1em,xshift = 1em] {ETH};

\draw[->] ([yshift=0em]investor.east) -- ([yshift=0em, xshift=-0.5em]servicebox.west)
  node[midway, yshift=1em, xshift=-1em] {\includegraphics[width=0.5cm]{coin.png}} 
  node[midway, yshift=1em,xshift = 1em] {LST};

\end{tikzpicture}
\caption{Examples of use cases for Liquid Staking Tokens (LSTs). This paper focuses specifically on the application of LSTs in liquidity provision in automated market makers (AMMs).}
\label{fig: 1}
\end{figure}

\quad Liquid staking effectively addresses the illiquidity of traditional staking, and its adoption has been extraordinary. For example, Lido Staked Ethereum (stETH), the largest LST, grew from about \$20 million in market value in January 2021 to over \$15 billion by mid-2023 \citep{scharnowski2025}, representing more than 30\% of all ETH staked on the network \citep{gogol2024}. In April 2023, the total value locked in liquid staking protocols collectively surpassed that of decentralized exchanges, making liquid staking the single largest category in DeFi by TVL \citep{jha2023}. This explosive growth underscores the significant economic role of LSP in the emerging Proof-of-Stake ecosystem.

\quad Despite their success, LSPs introduce complex decision problems for investors. By design, an LSP allows an investor to retain liquidity while staking: the investor can stake ETH via a liquid staking platform and receive an LST, which can then be used across DeFi protocols. A common strategy is to supply the LST and ETH as a pair to an automated market maker (AMM) pool (e.g. a decentralized exchange pool for ETH and its LST) to earn trading fees \citep{gogol2024, xiong2024}. 

\quad These opportunities come with important trade-offs. On one hand, liquid staking lowers barriers to participation and enables “yield stacking” -- simultaneously earning staking rewards and AMM fees. On the other hand, LSP yields are typically reduced by provider fees and may be more volatile \citep{scharnowski2025}. Moreover, LSTs often trade at a slight discount or premium relative to their underlying asset – the so-called liquid staking basis – which reflects factors such as the protocol’s reward rate, market volatility, and liquidity constraints. Providing liquidity in an AMM introduces further risks: the investor’s position is exposed to impermanent loss if the relative price of the two assets changes, and the earned fees might not fully compensate for this risk \citep{milionis2023, milionis2024, gogol2024liquid}. Empirical evidence indeed shows that a majority of LST liquidity providers have historically underperformed a simple buy-and-hold strategy – one study finds that about 66\% of LST liquidity positions yielded lower returns than just holding the LSTs \citep{xiong2024}. These observations highlight the need for a rigorous quantitative framework to determine optimal staking, liquidity provision, and exit policies. 
\smallskip
\setlength{\parindent}{0pt}\paragraph{\textbf{Main Contributions.}} 
\begin{enumerate}[itemsep = 3 pt]
    \item \textbf{Investor Decision Modeling.} In this paper, we develop a model to analyze the investor’s optimal entry and exit decisions in an Ethereum-based liquid staking and AMM setting. We consider a risk-sensitive investor who can hold ETH in three forms: (i) as unstaked ETH, (ii) as staked ETH via a liquid staking protocol (LSP), or (iii) as an liquidity provider (LP) share in an AMM pool composed of ETH and its corresponding liquid staking token (LST). The model captures the stochastic dynamics of token prices and the accumulation of staking rewards. This allows us to evaluate the investor's payoff and risks under different strategies.
    
    \item \textbf{Optimal Allocation for an Individual Investor and Insights for LSP and AMM Designers.}  We derive the investor’s optimal allocation among directly holding ETH, staking through an LSP, and providing liquidity to an AMM, which reflects the spirit of risk allocation. We establish the conditions under which an investor chooses to stake through a liquid staking protocol (LSP) and receive liquid staking tokens (LST), as well as the circumstances under which the investor enters an LST–ETH AMM pool. The results also provide insights for LSP and AMM designers. In particular, we identify a critical fee threshold for the AMM: only when this threshold is exceeded does the investor find it optimal to supply liquidity. The threshold result delivers a clear mechanism design insight for protocol designers, as it quantifies the minimum fee incentive required to attract a rational LST investor to provide liquidity.
        
    \item \textbf{Optimal Exit Timing for an Individual Investor}. We analyze the optimal timing to exit the AMM pool for an individual investor. After providing liquidity, the investor earns trading fees and staking rewards while exposed to stochastic price changes. We model the exit decision as an optimal stopping problem: at any moment, the investor can convert the LP share back into ETH and LST. We present the corresponding free-boundary problem formulation and restrict our attention to a class of stopping times defined by fixed price thresholds, namely $T = \inf\left\{t \geq 0: P_t \geq L\right\}$ or $T = \inf\left\{t \geq 0: P_t \leq L\right\}$, 
    where the investor exits once the LST price reaches a trigger threshold. We derive the expected objective value and apply Laplace transform techniques to characterize the optimal threshold $L$. Additionally, we decompose the investor's objective value and examine how each component evolves. Through both analytical derivations and numerical illustrations, we identify how these components jointly determine the investor's exit strategy.     
\end{enumerate} 

\smallskip
\setlength{\parindent}{0pt}
\paragraph{\textbf{Literature Review.}}

\quad Our research contributes to the growing literature on blockchain economics and staking. Prior studies have examined how wealth and stake distributions evolve under PoS consensus.  \citet{rocsu2021evolution} show that, in the absence of trading, stake shares follow a martingale process converging to stable long-run distributions. \citet{tang2022stabilitysharesproofstake} extends this line by analyzing different validator types, while \citet{tang2023trading} study optimal stake allocation using a continuous-time control framework.  \citet{tang2023tradingwealthevolutionproof} further surveys these dynamics and highlights the trade-offs faced by individual participants. Building on these foundations, our model connects PoS staking with DeFi activity by analyzing an investor who simultaneously holds LSTs and provides liquidity in an AMM. This integrated perspective bridges previously separate lines of research on staking and liquidity provision.

    \quad In addition, our research contributes to the literature on market microstructure, liquidity provision, and the role and optimal design of transaction fees in automated market makers (AMMs). \citet{capponi2021adoptionblockchainbaseddecentralizedexchanges} analyze the incentives for liquidity providers and show that they may incur losses when token exchange rates fluctuate. \citet{milionis2024} identify the main adverse selection cost incurred by liquidity providers (LPs) and propose measuring LP position losses relative to a rebalancing strategy, referred to as “loss-versus-rebalancing”. They later extend the model to incorporate the impact of trading fees on arbitrageurs' profits and demonstrate that the introduction of fees can be interpreted as a rescaling of time \citep{milionis2023}. We model AMMs following their framework and additionally account for the  opportunity costs and gains associated with liquid staking, while also examining the allocation and optimal stopping problem faced by an investor. Instead of the “loss-versus-rebalancing” approach, we adopt the conventional measure of impermanent loss, defined relative to a holding strategy. This choice is motivated by the fact that continuous and frictionless rebalancing, as required in their framework, may be impractical due to the unstaking process in liquid staking protocols.

    \quad 
    \cite{caodavid2025structuralmodel} conduct the first study to focus exclusively on optimal AMM fees and optimize fee structures by modeling a game among the platform, liquidity providers, and traders. Rather than assuming trading volume to be constant, they model it as a function of liquidity, fees, and price volatility, which influences the fees earned by liquidity providers. Therefore, the optimal fee they propose captures the interactions among these three parties. By contrast, our primary focus is on liquid staking, with the AMM component serving as one part of the overall framework. Our paper emphasizes a dual perspective: it examines both the dynamics of AMMs in the presence of liquid staking tokens (LSTs) or, equivalently, the application of LSTs within AMMs. This dual perspective is a direction not studied in previous literature. Within this perspective, we characterize transaction fees as incentives for liquidity providers and focus on investors’ decisions at time $t = 0$ (allocation problem) and at later times $t > 0$ (optimal stopping problem).


\quad     \citet{tang2023} address a mechanism design problem in the context of PoS consensus, proposing modifications to Ethereum’s transaction fee auction to ensure incentive compatibility and desirable equilibrium properties. \citet{bergault2024automatedmarketmakingcase} examine the exchange rate dynamics between two intrinsically linked cryptoassets and design an AMM model that exploits the the distinctive characteristics of this exchange rate, which may be applied to stablecoins and liquid staking tokens (LSTs). Our research contributes a complementary design insight in the DeFi context: beyond exchange rate dynamics, we incorporate additional features unique to LSTs, such as reward accumulation, the withdrawal process, and the risks associated with providing liquidity using LSTs. By identifying the fee threshold for LST liquidity provision, we inform how AMMs or protocol developers might adjust fee rates or introduce liquidity provision rewards to achieve desired participation outcomes. More broadly, our results demonstrate how tools from stochastic control and optimal stopping can be applied to DeFi mechanism design, yielding normative recommendations for platform policies.

\quad Our research is also closely related to emerging work on liquid staking tokens (LST) and DeFi markets. Recent studies have examined how LST prices, staking yields, and liquidity provision outcomes evolve under different conditions. \citet{scharnowski2025} analyze the sensitivity of LST discounts to market volatility and reward rates. \citet{xiong2024} highlight widespread underperformance among LST liquidity providers. While these works provide valuable empirical and descriptive insights, they do not offer a unified theoretical explanation of the observed behaviors. Our research builds a structural model that rationalizes these observed patterns and characterizes investor behavior through explicit trade-offs under uncertainty.

\medskip
{\bf Organization of the paper}: The remainder of the paper is organized as follows. Section \ref{sc: operational details} introduces the operational details of liquid staking and presents the model setup. In Section \ref{sc: optimal allocation}, we analyze the optimal allocation of ETH between direct holding, staking, and AMM liquidity provision, and propose minimal transaction fee functions that incentivize the investor to participate. Section \ref{sc: optimal time to exit} investigates the optimal time to exit the AMM pool. We present numerical findings in Section \ref{sc: numerical studies}. Finally, we conclude with Section \ref{sc: conclusion}.


\section{Operational Details and Model Formulation}
\label{sc: operational details}

\quad In this section, we present the operational details of liquid staking, and the formulation of our basic models.



\quad First, here is a list of some of the common notations used throughout the paper.
\begin{itemize}
    \item $\mathbb{N}_+$ denotes the set of positive integers, $\mathbb{R}$ denotes the set of real numbers, and $\mathbb{R}_+$ denotes the set of positive real numbers.
    \item $\left[ n\right]$ denotes the set $\{1,2,\ldots,n \}$.
    \item $\mathbb{P} \left( \cdot \right)$ denotes probability, $\mathbb{E} \left( \cdot \right)$ denotes expectation, and $\mathbb{E}^{s,x} \left( \cdot, \cdot\right)$ denotes the expectation conditional on the process starting at time $s$ and state $x$.
    \item $C^2 \left(\mathbb{R}^2 \right)$ is the set of functions that are twice continuously differentiable on $\mathbb{R}^2$.
\end{itemize}

\quad {\bf Liquid Staking Protocols (LSPs)} are smart contracts that enable users to stake their ETH and earn staking rewards while retaining liquidity. Examples of LSPs include Lido and RocketPool. These protocols are typically made up of validators, node operators, and a staking pool \citep{xiong2024}. Unlike traditional staking, where assets are locked, users receive liquid staking tokens (LSTs) upon depositing ETH into the pool. These LSTs allow users to withdraw their staked ETH and are transferable and tradable across various DeFi protocols, thereby maintaining liquidity while the underlying ETH remains staked and accrues rewards. For example, LSTs can be used to provide liquidity to decentralized exchange (DEX) pools or to trade within these pools \citep{gogol2024, xiong2024, xiong2025}. When users \textbf{withdraw} their ETH from the staking pool, the redemption is fulfilled directly if there is sufficient ETH available. Otherwise, the LSP must request node operators to unstake ETH, which may involve queuing and delays, thereby postponing the redemption process for users and potentially result in losses \citep{neuder2024optimizingexitqueuesproofofstake, gogol2024}. In this paper, the process of withdrawing ETH from the pool is also referred to as exiting the liquid staking protocol (LSP).


\quad A {\bf liquid staking token (LST)} is a tokenized representation of staked assets \citep{bitcoinlst,gogol2024}. Each LST is associated with two distinct prices: a \textbf{primary market price} and a \textbf{secondary market price} \citep{gogol2024, xiong2024}. The primary market price is determined within liquid staking protocols (LSPs) and is often referred to as the protocol or reserve value. At this price, users stake ETH and receive LSTs directly from the LSP. In contrast, the secondary market price reflects the market valuation of LSTs on centralized or decentralized exchanges, where LSTs can be traded. Discrepancies between the two prices may give rise to arbitrage opportunities \citep{gogol2024, xiong2024, bitcoinlst}. 

\quad Liquid staking tokens (LSTs) can be broadly classified into two categories: rebasing tokens and reward-bearing tokens. \textbf{Rebasing} tokens utilize a rebasing mechanism that maintains a constant 1:1 peg to ETH, with staking rewards reflected by adjusting the number of tokens held in users' wallets \citep{xiong2024, gogol2024, lidosteth}. An example of a rebasing token is stETH, which is issued and managed by Lido. In contrast, \textbf{reward-bearing} LSTs accrue staking rewards through increase in their token values. Consequently, their price relative to ETH rises over time, while the number of tokens held remains unchanged \citep{taguslabs, gogol2024, xiong2024}. An example of a reward-bearing token is rETH, issued by RocketPool.

\smallskip
\setlength{\parindent}{0pt}\paragraph{\textbf{Price Dynamics of Liquid Staking Tokens (LST)}}
Let $P_t$ denote the price of the liquid staking token (LST) in terms of ETH. It follows a geometric Brownian motion given by the stochastic differential equation
\begin{equation*}
    \frac{dP_t}{P_t} = g \, dt+ \sigma \, dB_t,
\end{equation*}
where $g \geq 0$ is the drift term, $\sigma \in \mathbb{R}_+$ represents the volatility, and $B_t$ is a standard Brownian motion. Let $P_0$ denote the initial price. For rebasing tokens (e.g., stETH on the Lido platform), $P_0 = 1$ and $g = 0$; for reward-bearing tokens (e.g., rETH on Rocket Pool), $P_0 > 1$ and $g > 0$.

\quad Let $D_t$ denote the discount factor that applies when the investor exits the liquid staking protocol (LSP). (The terms ``exit the LSP'' and ``withdraw'' are used interchangeably.) For $t \geq 0$, $D_t = e^{-mt}$ where $m > 0$ is the associated discount parameter. At time $t$, the realized price upon withdrawing is thus $D_t P_t$. 

\begin{remark}
    The intuition behind $D_t = e^{-mt}$ is as follows: the longer the investor stakes in the liquid staking protocol (LSP), the greater the staking rewards she accrues. Consequently, when a larger quantity of tokens is withdrawn, the discounting factor becomes more significant.
\end{remark}

\quad Let $\rho$ denote the discount rate for the investor when we discount the investor's payoff at time $t$ to its present value. Therefore, the present value of the realized price upon unstaking at time $t$ is $e^{-\rho t} D_tP_t$.



\smallskip

\setlength{\parindent}{0pt}\paragraph{\textbf{The Investor's Decision-Making}}

Consider an investor who begins with 1 ETH and deposits it into a liquid staking protocol. If the investor exits the LSP at time $t$, the ETH will have grown to a value of $\frac{1}{P_0}e^{rt}D_t P_t$. The investor is incentivized to stake through the protocol only if the expected present value of the staked ETH and the accumulated rewards exceeds the initial value of 1 ETH, that is, when
\begin{equation*}
\frac{1}{P_0}\mathbb{E} \left[e^{\left(-\rho+r\right) t} D_t P_t\right] > 1,
\end{equation*}
where $\rho$ is the investor’s discount rate, and $P_0$ is the initial LST price. This inequality ensures that, in expectation, the present value of staking exceeds that of simply holding ETH. The condition above simplifies to the following requirement:
\begin{equation}
   r+g-\rho-m > \ln P_0 > 0. \label{assump: stake via LSP}
\end{equation}
\quad This condition essentially implies that the difference between the reward rate (with the price growth rate $g$ interpreted as an alternative form of staking reward) and the two discount rates, $m$ and $\rho$, must be sufficiently large to incentivize investors to stake via a liquid staking protocol (LSP). This can potentially be achieved by improving validator or node operator performance to enhance the reward rate, or by reducing the risks and potential delays users may encounter when withdrawing their ETH. Throughout the remainder of this paper, we assume that this inequality always holds.

\quad After depositing into the LSP, the investor decides whether to provide liquidity to an automated market maker (AMM) pool using the LST and the remaining ETH. To be more precise, the investor determines how to allocate the 1 unit of ETH between direct holding and staking. She stakes $\left(1-a\right)$ units of ETH via an LSP, receiving $\frac{1-a}{P_0}$ units of LST. The staked portion accumulates staking rewards equal to $\left(e^{rt}-1\right) \cdot \frac{1-a}{P_0}$ units of LST, where $r > 0$ for rebasing tokens and $r = 0$ for reward-bearing tokens. The remaining $a$ units of ETH is held. The investor then chooses to deposit $x$ units of LST out of the $\frac{1-a}{P_0}$ received and $y$ units of ETH out of the $a$ held to an automated market maker (AMM) pool to provide liquidity. At a later time $t$, the investor plans to exit the pool and realize returns. Figure \ref{fig: flowchart} illustrates the process of liquid staking and liquidity provision. We would like to explore the following two questions:
\begin{enumerate}[itemsep = 3 pt]
    \item \textbf{Optimal Allocation}: For \textbf{an individual investor}, what is the optimal allocation of the initial ETH between staking and direct holding, and subsequently between AMM liquidity provision and holding? Based on these results, what are the necessary and sufficient conditions for \textbf{protocol designers} to effectively incentivize investor participation?
    \item \textbf{Optimal Time to Exit}: For \textbf{an individual investor}, when is the optimal time to exit the AMM pool in order to maximize expected returns?
\end{enumerate}

We address these questions sequentially in Sections \ref{sc: optimal allocation} and \ref{sc: optimal time to exit}. 

\begin{figure}[h] 
\centering
\begin{tikzpicture}[
  node distance=2cm and 1.6cm,
  box/.style={draw, rectangle, minimum width=1.5cm, minimum height=1cm, align=center},
  >=Stealth
]

\node[box] (eth) {1 ETH};
\node[box, right=of eth, yshift=1cm] (1alpha) {{\color{red}$1 -$  $a$} ETH};
\node[box, below=of 1alpha, yshift = 1cm, xshift = -0.3cm] (alpha) {{\color{red} $a$} ETH};

\node[box, right=of 1alpha, xshift = - 0.3cm] (staking) {Liquid Staking \\ Protocol};
\node[box, right=of alpha] (amm) {LST-ETH AMM};

\node[box, right=of staking, yshift=1cm, xshift = 0.3cm] (reward) {Staking Rewards: \\ {\color{red} $\frac{\left(e^{rt}-1 \right) \left( 1-a\right)}{P_0}$} LST };
\node[box, below=of reward, yshift = 1cm, xshift = 0.05cm] (lst) { {\color{red} $\frac{1-a}{P_0}$} LST};

\draw[->] (eth) -- (1alpha);
\draw[->] (eth) -- (alpha);

\draw[->] (1alpha) -- (staking);
\draw[->] (staking) -- (reward);
\draw[->] (staking) -- (lst);

\draw[->] (alpha) -- (amm);
\draw[->] (alpha) -- node[midway, above] {{\color{red} $y$} ETH} (amm);

\draw[->] (lst) -- node[midway, above] {{\color{red} $x$} LST} (amm);

\end{tikzpicture}
\caption{Liquid Staking and Liquidity Provision. All quantities refer to token units (e.g., “1 ETH” denotes one unit of the ETH token). }
\label{fig: flowchart}
\end{figure}


\section{Optimal Allocation}
\label{sc: optimal allocation}

\quad In this section, we analyze the investor's optimal allocation strategy, given that she decides to stake through an LSP. Section \ref{sc: CPMM} analyzes how a constant product market maker (CPMM) rebalances an investor's portfolio. Section \ref{sc: Liquidity provision condition} presents the condition an investor must satisfy to provide liquidity to an automated market maker (AMM). Section \ref{sc: problem formulation} formulates the optimal allocation problem, and Section \ref{sc: optimal solutions} analyzes the optimal solutions and establishes a lower bound on AMM fees, which is a necessary and sufficient condition, offering insights to AMM designers.


\subsection{Constant Product Market Maker (CPMM)}
\label{sc: CPMM}

The constant product market maker (CPMM), such as Uniswap, is currently the most widely adopted type of AMM. As discussed in Remark \ref{rmk: CPMM}, CPMMs also ensure the fair treatment of investors. Accordingly, we focus on CPMMs in this paper. Specifically, we consider a CPMM with initial reserves of $X$ units of LST and $Y$ units of ETH. At time $t$, the pool is rebalanced to $U^{\star}$ units of LST and $V^{\star}$ units of ETH, where $U^{\star}$ and $V^{\star}$ is obtained by solving the following pool value problem \citep{milionis2024}:

\begin{equation}
\begin{aligned}
    V\left(P_t;L\right) =\min_{U,V \in \mathbb{R}^+} & D_tP_t U + V, \\
\text{subject to \,} & UV = L,
\end{aligned}
\label{prob: pool value problem}
\end{equation}
where $L = XY$. By substituting $u = \frac{L}{v}$ into the objective function, we obtain the optimal solution to Problem \ref{prob: pool value problem}:
\begin{equation}
    U^\star = \sqrt{\frac{L}{D_t P_t}}, \quad V^\star = \sqrt{D_t P_t L}, \label{eq: opt sol for pool}
\end{equation}
with the optimal objective value given by $V\left(P_t;L\right) = 2 \sqrt{L D_tP_t}$.


\begin{remark}
    This corresponds to identifying the point on the curve $uv = L$ where the marginal exchange rate, given by $\frac{v}{u}$, equals the external realized price $D_t P_t$. (A more detailed discussion of this exchange rate is provided below.) This indicates that the price within the AMM pool should align with the price outside the AMM pool, which is due to the existence of arbitrageurs \citep{milionis2024}.
\end{remark}

\quad Now, suppose an investor deposits $x$ units of LST and $y$ units of ETH into the pool. In return, she receives liquidity provider tokens representing her proportional share of the pool’s total value, which may later be redeemed for a corresponding share of the pool’s assets \citep{cryptopedia2025lptokens, coinmarketcap, xrpledger2025}. Here, the deposited amounts must satisfy $x = \lambda X$ and $y = \lambda Y$, where $0 < \lambda < 1$ represents the investor’s share of the pool. For a detailed justification, refer to Proposition \ref{prop:liquidity provision condition} and Eq.~\eqref{assump: liquidity provision condition}. The next lemma and proposition characterize the rebalancing of a single investor's deposits within the pool. 
\begin{lemma}
Let $x = \lambda X$ and $y = \lambda Y$, where $\lambda > 0$. Let $\left(U^\star, V^\star \right)$ denote the optimal solution to Problem \ref{prob: pool value problem all} (as given in Eq. \eqref{eq: opt sol for pool}), and $\left(u^\star, v^\star \right)$ denote the optimal solution to Problem \ref{prob: pool value problem single}:
    \begin{equation}
\begin{aligned}
    V\left(P_t; L = XY \right) = \min_{U,V \in \mathbb{R}^+} & D_tP_t U + V, \\
\text{subject to \,} & UV = L =XY,
\end{aligned}
\label{prob: pool value problem all}
\end{equation}

    \begin{equation}
\begin{aligned}
    V\left(P_t; L^\prime = xy \right) = \min_{u,v \in \mathbb{R}^+} & D_tP_t u + v, \\
\text{subject to \,} & uv = L^\prime =xy,
\end{aligned}
\label{prob: pool value problem single}
\end{equation}
Then, we have
\begin{equation}
    \begin{aligned}
        u^\star &= \lambda U^\star, \\
        v^\star &= \lambda V^\star, \\
        V\left(P_t;L^\prime \right) &= \lambda V\left(P_t;L \right).
    \end{aligned}
    \label{eq: opt sols and value}
\end{equation}
\label{lem:lmda pool}
\end{lemma}

The proof of Lemma \ref{lem:lmda pool} is provided in Appendix \ref{appx: proof of results in sc 3}.

\begin{proposition}
Starting with $x$ units of LST and $y$ units of ETH, an investor's portfolio is rebalanced to $u^\star = \sqrt{\frac{xy}{D_tP_t}
}$ and $v^\star = \sqrt{D_t P_t xy}$ at time $t$, and the portfolio value is equal to $2\sqrt{xy D_tP_t}$.
\label{prop: single investor rebalancing}
\end{proposition}

\begin{proof}
    The results follow immediately from Eq. \eqref{eq: opt sol for pool} and Lemma \ref{lem:lmda pool}.
\end{proof}

\begin{remark}
\label{rmk: CPMM}
   Lemma \ref{lem:lmda pool} and Proposition \ref{prop: single investor rebalancing} yield several important insights. First, they show that for a single investor, providing liquidity to an AMM and receiving a proportional share of the pool’s value is equivalent to solving the pool value problem with the invariant function parameter $L$ set to be the product of the investor’s initial deposits. In other words, the value of a liquidity provider’s share is not affected by the entry or exit of the other liquidity providers, as noted in \citep{zhu2025optimalexiting}.
   
   \quad Second, they show that the constant product market maker (CPMM) treats all investors equivalently: each investor’s terminal portfolio $\left( u^\star, v^\star \right)$ depends solely on their own share of the pool, independent of other factors. This fairness property may help explain the widespread adoption of CPMMs, which is why we focus on CPMMs in this study.
\end{remark}

\subsection{Liquidity Provision Condition} 

\label{sc: Liquidity provision condition}

Under the CPMM invariant $XY = L$, the ETH reserve can be expressed as a function of the LST reserve: $Y = h(X) = \frac{L}{X}$, where $X$ and $Y$ denote the reserves of LST and ETH, respectively. The marginal exchange rate, defined as the rate at which a liquidity trader receives ETH when selling 1 unit of LST, is given by 
\begin{equation*}
    Z = - h^{\prime}(X) = \frac{Y}{X},
\end{equation*}
according to \citet{cartea2025}.
 Let $X_0$ and $Y_0$ denote the pool’s initial reserves of LST and ETH prior to the investor’s liquidity provision. 
 The investor's provision $\left(x,y \right)$ must satisfy the liquidity provision condition \citep{cartea2025}:
\begin{equation}
    \frac{y + Y_0}{x+X_0} = \frac{Y_0}{X_0}. \label{assump: liquidity provision condition}
\end{equation}
This ensures that liquidity provision does not change the marginal rate, which is a distinctive characteristic of AMMs \citep{angeris2020, cartea2025}.
We further assume that, at time zero (i.e., prior to the investor’s entry), the pool is already rebalanced with the external market:
\begin{equation}
    \frac{Y_0}{X_0} = P_0. \label{assump: pool already rebalanced}
\end{equation}
\begin{proposition}
    When an investor deposits $x$ LST and $y$ ETH into the pool, they must satisfy:
\begin{equation}
    \frac{y}{x} = P_0. \label{constr: y=P_0x}
\end{equation}
\label{prop:liquidity provision condition}
\end{proposition}
\begin{proof}
Combining conditions~\ref{assump: liquidity provision condition} and~\ref{assump: pool already rebalanced}, the result follows immediately.
\end{proof}

\begin{remark}
The liquidity provision condition requires that the value of the investor's supplied LST and ETH be equal, since $P_0x=y$.
\label{rmk: CPMM equal value}
\end{remark}

\subsection{Problem Formulation}
If the investor exits the AMM and the LSP at time $t$, the expected present value of her portfolio becomes
\label{sc: problem formulation}
\begin{equation*}
    e^{- \rho t} \  \mathbb{E} \left[ 2 \sqrt{xyD_tP_t} + D_tP_t \left(\frac{1-a}{P_0}-x\right) + \left(a-y \right) + D_tP_t  \frac{\left(e^{rt}-1\right)\left(1-a\right)}{P_0} \right],
\end{equation*}
where $\rho$ denotes the investor's discount rate. Here, $2 \sqrt{xyD_tP_t}$ represents the value of the tokens contributed by the investor to the AMM pool after rebalancing (as described in Section \ref{sc: CPMM}). The term $D_tP_t \left( \frac{1-a}{P_0} -x \right)$ denotes the value of the LST retained outside the pool, and $\left(a-y\right)$ corresponds to the remaining ETH not supplied to the AMM. Finally, $D_tP_t  \frac{\left(e^{rt}-1\right)\left(1-a\right)}{P_0}$ is the value of the staking rewards earned through the LSP.

\quad Additionally, assume that the investor earns transaction fees at rate $\phi_t$ per unit of LST deposited. Then, by time $t$, the investor has accrued a total of $x \cdot \int_{0}^t \phi_s e^{-\rho s} ds$ transaction fees discounted to the present value.

\begin{remark}
    We model the fee rate in units of deposited LST, in line with \citet{zhu2025optimalexiting}. However, while \citet{zhu2025optimalexiting} assumes a constant fee rate, we allow the rate to vary with time 
$t$, which can be seen as a natural generalization of their assumption.
\end{remark}

\quad To determine the optimal allocation strategy, we consider the following optimization problem:
\begin{equation*}
    \begin{aligned}
        \max \ & e^{-\rho t} \ \mathbb{E} \left[ 2 \sqrt{xyD_tP_t} + D_tP_t \left(\frac{1-a}{P_0}-x\right) + \left(a-y \right) + D_tP_t  \frac{\left(e^{rt}-1\right)\left(1-a\right)}{P_0} \right] \\
        & + x\int_{0}^t \phi_s e^{-\rho s} ds , \\
        \text{subject to } \ & 0 \leq y \leq a, \\
        & 0 \leq x \leq \frac{1-a}{P_0}, \\
        & 0\leq a \leq 1, \\
        & \frac{y}{x} = P_0.
    \end{aligned}
\end{equation*}

The objective function is defined as the sum of the investor’s expected  portfolio value at time $t$ and the transaction fees she earned, both discounted to their present values. The first three constraints follow naturally from their definitions, and the last constraint corresponds to the liquidity provision condition, as detailed in Section \ref{sc: Liquidity provision condition}.

\quad Substituting $y = P_0 x$, the problem simplifies to:
\begin{equation}
    \begin{aligned}
        \max \ &\left\{-e^{-\rho t}\mathbb{E} \left[ \left( \sqrt{D_tP_t} - \sqrt{P_0} \right)^2 \right] + \int_{0}^t \phi_s e^{-\rho s} ds \right\}x \\& - e^{-\rho t} \,  \left[\frac{\mathbb{E}\left(D_tP_t\right) e^{rt}}{P_0} - 1 \right]a + e^{-\rho  t} \frac{\mathbb{E}\left(D_tP_t\right) e^{rt}}{P_0}, \\
        \text{subject to} \ & P_0 x \leq a \leq 1-P_0x, \\
        &0 \leq x \leq \frac{1}{2P_0}.
    \end{aligned} \label{prob: double max prob}
\end{equation}

\subsection{Optimal Solutions}
\label{sc: optimal solutions}
We solve this problem by first determining the optimal value of $a$ for a fixed $x$, and then solving for the optimal $x$. To complete the analysis, we impose a natural tie-breaking assumption: 
\begin{assumption}
    If the investor is indifferent between taking a positive action and remaining inactive, they opt to do nothing. That is, if any $x^\star \in \left[0, \frac{1}{2P_0} \right]$ is optimal, then the investor chooses $x^\star = 0$; similarly, if any $a^\star \in \left[0,1 \right]$ is optimal, then the investor sets $a^\star = 0$. \label{assump: tie-breaking}
\end{assumption}

\begin{remark}
    This assumption is motivated by the fact that, in practice, investor participation without clear gains is often undesirable due to real-world frictions, such as transaction costs, effort, and risk.
\end{remark}



\quad The optimal value for $x^\star$ depends on the cumulative transaction fees, $\int_0^t \phi_s e^{-\rho s} ds$. By examining different values of the cumulative fees, we obtain the following proposition and theorem, which characterize the optimal solution and establish the necessary and sufficient condition for Problem (\ref{prob: double max prob}) to admit a positive optimal solution. Proofs of these results are presented in Appendix \ref{appx: proof of results in sc 3}.

\begin{proposition}
$x^\star = \frac{1}{2P_0}$, $a^\star = \frac{1}{2}$, and $y^\star = P_0 x^\star = \frac{1}{2}$.
\label{prop: optimal allocation solution}
\end{proposition}


\begin{remark}
    This proposition offers several important insights. First, if an investor retains a portion of ETH instead of staking the full amount of 1 ETH, then the optimal strategy is to allocate the entire unstaked portion to the AMM pool. This follows from the fact that $y^\star = a^\star = \frac{1}{2}$. The rationale is that if transaction fees provide sufficient incentive for participation, then there's no reason to withhold any remaining ETH. 
    
    \quad Second, $a^\star = \frac{1}{2}$ can be interpreted as an equal-risk allocation strategy: half of the ETH is staked through the LSP, and the other half is allocated to provide liquidity in the AMM pool. This strategy is distinctive to CPMMs, as CPMMs require that the supplied LST and ETH be of equal value (as noted in Remark \ref{rmk: CPMM equal value}).
    
    \label{rmk: opt allocation risk}
\end{remark}

\begin{theorem}
    Let 
    \begin{align*}  
    \Phi \left( t\right) &= P_0 \left[ e^{\left(g-m-\rho \right)t} \left( e^{rt}+1 \right) - 2  e^{\left(\frac{1}{2}g - \frac{1}{8} \sigma^2 - \frac{1}{2}m - \rho \right)t}  \right].
    \end{align*}
$x^\star = \frac{1}{2P_0}$ and $a^\star = \frac{1}{2}$ if and only if:
\begin{equation*}
    \int_{0}^t \phi_s e^{-\rho s} ds > \Phi \left( t\right).
\end{equation*} \label{thm: bounds for fees} 
\end{theorem}
\vspace{-1em}
This theorem establishes a lower bound on the transaction fees required to incentivize liquidity provision in the AMM. The first term of $\Phi_t$, given by $e^{-\rho t} x^\star\mathbb{E} \left[  \left(\sqrt{D_tP_t} - P_0 \right)^2 \right]$, is to offset the impermanent loss. The second term, $e^{-\rho t} a^\star \left[ \frac{\mathbb{E} \left( D_tP_t\right)e^{rt}}{P_0} - 1 \right]$, is to compensate for the opportunity cost faced by an investor who chooses to stake stake $1-a^\star = \frac{1}{2}$ ETH in the AMM rather than staking the full amount of 1 ETH through an LSP. Unlike other tokens, where the transaction fees provided by the AMM only accounts for impermanent loss, in this case, the fees must also compensate for the opportunity cost associated with staking. This feature is unique to liquid staking tokens.


\quad Automated Market Makers (AMMs) seek to minimize the transaction fees paid to liquidity providers (LPs) while maintaining sufficient incentives for LPs to supply liquidity, as lower fees can also enhance incentives for traders. Therefore, we aim to identify the lower bounds of transaction fees and to characterize fee structures that achieve these bounds. Specifically, we investigate functions $\phi_t$ for $t \geq 0$ that satisfy the following condition:
\begin{equation*}
    \int_{0}^t \phi_s e^{-\rho s} ds =  \Phi \left( t\right).
\end{equation*}

\begin{proposition}

There exists a function $\phi_t$ for $t \geq 0$ that attains the lower bounds established in Theorem~\ref{thm: bounds for fees}. This function is given explicitly by
 \begin{align*}
    \phi_t = P_0 \left[ 2\left(g-m-\rho \right) e^{\left(g-m-\rho \right)t }  - \left(g - \frac{1}{4} \sigma^2 - m -2 \rho \right) e^{\left(\frac{1}{2}g - \frac{1}{8} \sigma^2 - \frac{1}{2}m - \rho \right)t} \right],
    \end{align*}
for reward-bearing liquid staking tokens, where $g > 0$ and $r=0$; and
\begin{align*}
    \phi_t = P_0 \left[ \left(r-m-\rho \right) e^{\left(r-m-\rho \right)t } -\left( m+\rho\right) e^{\left(-m-\rho \right)t} + \left(  \frac{1}{4} \sigma^2 + m +2 \rho \right) e^{\left( - \frac{1}{8} \sigma^2 - \frac{1}{2}m - \rho \right)t} \right],
\end{align*}
for rebasing liquid staking tokens, where $g = 0$ and $r>0$.
\end{proposition}
\label{fee structure}





\quad Although transaction fees are necessary to incentivize liquidity providers to lock in their assets, it is important to prevent them from growing unrealistically. To address this, we set an upper bound on the fees. Secifically, we set the cumulative transaction fee for both reward-bearing and rebasing tokens to be at most $2P_0K$, where $K$ is a constant. Here, we denote the upper limit by $2P_0K$ instead of $K$ for the ease of notation in the following sections.

\section{Optimal Exit Time}
\label{sc: optimal time to exit}
\quad In this section, we analyze the optimal time for an investor to exit the AMM pool. Section \ref{sc: optimal exit problem formulation} formulates the problem in the framework of a free boundary problem. Seciond \ref{sc: fixed stopping levels} considers fixed stopping levels. Section \ref{sc: without transaction fees} analyzes the properties of the objective function and the optimal solution in the absence of transaction fees, and Section \ref{sc: with transaction fees} explores the characteristics of the problem with transaction fees.

\subsection{Formulation: Free Boundary Problem} \label{sc: optimal exit problem formulation}
Given the protocol design specified in the preceding sections, we have $a^\star = \frac{1}{2}$ and $x^\star = \frac{1}{2P_0}$. Accordingly, the investor’s total portfolio value at time $t$, together with the accrued fees, is given by
\begin{align*}
     S\left(t, P_t\right) = \frac{1}{2P_0} \int_{0}^t \phi_s e^{-\rho s} ds + e^{-\rho  t} \left[ \frac{\left(e^{rt} - 1\right)D_t P_t}{2 P_0} + \sqrt{\frac{D_t P_t}{P_0}} \right].
\end{align*}
Compared to pure staking through an LSP, the present value of the net profit from providing liquidity to an AMM pool at time $t$ is:
\begin{align}
    W \left(t, P_t\right) &=  S\left(t, P_t\right) - e^{-\rho t} \frac{e^{rt} D_t P_t}{P_0}, \notag \\
    &=\frac{1}{2 P_0} \int_{0}^t \phi_s e^{-\rho s}  ds +  e^{-\rho t} \,  \left( - \frac{e^{rt} D_t P_t}{2 P_0} + \sqrt{\frac{D_t P_t}{P_0}} - \frac{ D_t P_t}{2 P_0} \right).
    \label{eq: W(t,Pt)}
\end{align}



We consider the optimal stopping problem with the value function
\begin{equation}
    V_{\star}\left(P \right) = \sup_{\tau \in \mathcal{T}} \ \mathbb{E} \left[W\left(\tau, P_{\tau}\right)\right],
    \label{eq: OST}
\end{equation}
where $\mathcal{T}$ denotes the family of stopping times that we are interested in. 

\quad Let $\tilde{P_t} = \dfrac{P_t}{P_0}$. 
We have
\begin{equation*}
    d\tilde{P_t} = g\tilde{P_t} \, dt+ \sigma\tilde{P_t} \, dB_t.
\end{equation*}


Substituting the transaction fee specified in Theorem~\ref{thm: bounds for fees}, together with the imposed upper bound, into Eq.~\eqref{eq: W(t,Pt)}, we define the corresponding reward function as follows:
\begin{equation}
    Z\left(t,x \right) = \min \left\{\frac{1}{2}e^{\left(g-m-\rho \right)t} \left( e^{rt}+1 \right) -   e^{\left(\frac{1}{2}g - \frac{1}{8} \sigma^2 - \frac{1}{2}m - \rho \right)t}, K \right\} + e^{-\rho t} \left( - \frac{e^{(r-m)t} x}{2 } + e^{-\frac{1}{2}mt} \sqrt{x} - \frac{ e^{-mt} x}{2 } \right).
    \label{eq: Z(t,x)}
\end{equation}

The objective is to determine a stopping time $\tau$ that maximizes
\begin{equation*}
    \mathbb{E}^{s,x} \left[ Z \left(\tau, \tilde{P}_{\tau} \right) \right],
\end{equation*}
where $\mathbb{E}^{s,x} \left( \cdot, \cdot\right)$ denotes the expectation conditional on the process starting at time $s$ and state $x$. In particular, setting $s = 0$ yields the optimal solution to our problem. 

\quad Let $Y_t = \left( s+t, \tilde{P}_t \right)$. The characteristic operator $\hat{\mathcal{A}}$ associated with $Y_t$ is given by
\begin{equation*}
    \hat{\mathcal{A}} f\left(s,x \right) = \frac{\partial f}{\partial s} + gx\frac{\partial f}{\partial x} + \frac{1}{2} \sigma^2 x^2 \frac{\partial^2 f}{\partial x^2}, \quad f \in C^2 \left(\mathbb{R}^2 \right).
\end{equation*}

Let $D$ denote the continuation region, and define the subset
\begin{equation*}
    U \coloneqq \left\{  \left(s,x \right): \hat{\mathcal{A}} Z\left(s,x \right) > 0\right\} \subset D.
\end{equation*}
The set $\left\{x: x \in U \right\}$ takes the form $B_1 \left( s\right)<x<B_2\left( s\right)$, where $B_1, B_2$ are functions of $s$. Accordingly, we consider the optimal stopping time of the form $\tau^\star = \inf \left\{ t \geq 0: B_1 \left(t\right) \leq \tilde{P_t} \leq B_2 \left(t \right)  \right\}$.
Let $V_{\star} \left( s,x\right) = \mathbb{E}^{s,x} \left[ Z \left(\tau^\star, \tilde{P}_{\tau^\star} \right) \right]$. The value function $V_\star$ and the optimal stopping boundary $B_1\left(s\right)$ and $B_2\left(s\right)$ can be obtained by solving the following free boundary problem \citep{oksendal1998,Peskir2006}:


\begin{equation}
\begin{aligned}
    \hat{\mathcal{A}} V_{\star} &= 0, \quad \text{ for } B_1 \left(s\right) < x <B_2 \left(s\right),\\
    V_{\star} \left(s, B_1 \left(s\right) \right) &= Z\left(s, B_1 \left( s\right) \right),\\
    V_{\star} \left(s, B_2 \left(s\right) \right) &= Z\left(s, B_2 \left( s\right) \right),\\
    \frac{\partial V_\star \left( s,x \right
    )}{\partial x} \left|_{x = B_1 (s)^+} \right. &= \frac{\partial Z \left(s,x \right)}{\partial x} \left|_{x = B_1 (s)^+} \right.,\\
    \frac{\partial V_\star \left( s,x \right
    )}{\partial x} \left|_{x = B_2 (s)^-} \right. &= \frac{\partial Z \left(s,x \right)}{\partial x} \left|_{x = B_2 (s)^-} \right.,
\end{aligned}
\label{prob: free boundary}
\end{equation}

In general, solving the free-boundary problem \eqref{prob: free boundary} is analytically intractable. We therefore begin by restricting the analysis to stopping times with fixed stopping levels. We then use numerical experiments to investigate the free-boundary solution and examine how it differs from the fixed-level strategy.



\subsection{Fixed Stopping Levels}
\label{sc: fixed stopping levels}
Define $\tilde{L} = \frac{L}{P_0}$, $ d = \frac{\sigma}{2} - \frac{g}{\sigma}$, and $c = \frac{1}{\sigma} \ln \tilde{L}$.
The stopping time $\tau = T_{c,d}$ is defined as
\begin{equation*}
    T_{c,d}=\inf\left\{t \geq 0: P_t \geq L \right\} = \inf\left\{t \geq 0: B_t = c  + dt \right\},
\end{equation*} where $L > P_0$ and $c > 0$, or
\begin{equation*}
    T_{c,d}=\inf\left\{t \geq 0: P_t \leq L \right\} = \inf\left\{t \geq 0: B_t = c  + dt \right\},
\end{equation*}
where $L < P_0$ and $c < 0$. The set of stopping times we are interested in is defined as $\mathcal{T} = \left\{T_{c,d}:c \in \mathbb{R}, c\neq 0 \right\}$. According to \citet{10.3150/bj/1068129008}, the Laplace transform of the stopping time is given as
\begin{align*}
    \mathbb{E} \left( e^{-\lambda T_{c,d}}\right) =\begin{cases}
        \exp \left( -c\left( d + \sqrt{d^2 + 2\lambda}\right)\right), & \text{ when } c  > 0,\\
        \exp \left( -c\left( d - \sqrt{d^2 + 2\lambda}\right)\right), & \text{ when } c  < 0.
    \end{cases} 
\end{align*}
Since $\tilde{L} = e^{\sigma c}$, when $c>0$, our objective value can be expressed as
\begin{equation}
\begin{aligned}
    \mathbb{E} \left[ W 
    \left(T_{c,d}, P_{T_{c,d}} \right)
    \right] &=\min \left\{\frac{1}{2} e^{-c\left( d + \sqrt{d^2 - 2\left(g+r-m-\rho\right)}\right)} + \frac{1}{2} e^{-c\left( d + \sqrt{d^2 - 2\left(g-m-\rho\right)}\right)} \right.\\
    & \left. \quad \quad \quad \ \  - e^{-c\left( d + \sqrt{d^2 - \left(g - \frac{1}{4} \sigma^2 - m - 2\rho\right)}\right)}, K \right\} \\
    & \quad +  e^{-c\left( -\frac{1}{2} \sigma+ d + \sqrt{d^2 + m + 2\rho}\right)} -\frac{1}{2} e^{-c\left( - \sigma+ d + \sqrt{d^2 - 2\left(r-\rho-m\right)}\right)} \\
    & \quad - \frac{1}{2}  e^{-c\left( -\sigma+ d + \sqrt{d^2 + 2\left(\rho+m\right)}\right)}. 
\end{aligned}
\label{eq: obj value for OST c > 0}
\end{equation}
Similarly, when $c < 0$, 
\begin{equation}
\begin{aligned}
    \mathbb{E} \left[ W 
    \left(T_{c,d}, P_{T_{c,d}} \right)
    \right] &=\min \left\{\frac{1}{2} e^{-c\left( d - \sqrt{d^2 - 2\left(g+r-m-\rho\right)}\right)} + \frac{1}{2} e^{-c\left( d - \sqrt{d^2 - 2\left(g-m-\rho\right)}\right)} \right.\\
    & \left. \quad \quad \quad \ \  - e^{-c\left( d - \sqrt{d^2 - \left(g - \frac{1}{4} \sigma^2 - m - 2\rho\right)}\right)}, K \right\} \\
    & \quad +  e^{-c\left( -\frac{1}{2} \sigma+ d - \sqrt{d^2 + m + 2\rho}\right)} -\frac{1}{2} e^{-c\left( - \sigma+ d - \sqrt{d^2 - 2\left(r-\rho-m\right)}\right)} \\
    & \quad - \frac{1}{2}  e^{-c\left( -\sigma+ d - \sqrt{d^2 + 2\left(\rho+m\right)}\right)}. 
\end{aligned}
\label{eq: obj value for OST c < 0}
\end{equation}
\quad We make the following assumptions.

\begin{assumption}
    \begin{align*}
        \left( \frac{\sigma}{2} - \frac{g}{\sigma} \right)^2 - 2\left( g+r-m-\rho \right) &> 0, \\
         \left( \frac{\sigma}{2} - \frac{g}{\sigma} \right)^2 - \left( g - \frac{1}{4} \sigma^2 - m -2\rho\right) &> 0, \\
         d = \frac{\sigma}{2} - \frac{g}{\sigma} & >0, \\
        \frac{1}{4} \sigma^2 - m &> 0.
    \end{align*}
    \label{assump: range for paras}
\end{assumption}

    These conditions ensure that the investor’s problem is well-posed and admits non-degenerate solutions. In particular, they reflect the fact that the reward rate $r$ for rebasing tokens, and the price growth rate $g$ for reward-bearing tokens are moderate relative to the uncertainty in the market, represented by $\sigma$. This is consistent with the high volatility typically observed in crypto markets, where token prices can fluctuate significantly over short horizons. Moreover, the discounting factor $m$ captures the potential delay when the liquid staking protocol (LSP) lacks sufficient ETH for immediate withdrawal and must unstake from the chain, as discussed in Section \ref{sc: operational details}. The parameter $\rho$ is used to discount future payoffs to present values. Both of them should be modest compared to volatility.



\quad Let $T_i^+$ denote the $i$-th term term in Eq. (\ref{eq: obj value for OST c > 0}), and $T_i^-$ denote the $i$-th term in Eq. (\ref{eq: obj value for OST c < 0}) (assuming there's no upper limit $K$), $i = \left[ 6\right]$.The explicit expressions of $T_i^+$ and $T_i^-$ are provided in Appendix~\ref{appx: Ti}.

\quad Although from Section \ref{sc: optimal allocation}, the net profit for the investor from providing liquidity to an AMM pool without transaction fees, relative to pure staking, is always negative under Assumption \ref{assump: stake via LSP}, as given by
\begin{align*}
    -e^{-\rho t}x^\star \mathbb{E} \left[ \left( \sqrt{D_tP_t} - \sqrt{P_0} \right)^2 \right] - e^{-\rho t} \,\left[\frac{\mathbb{E}\left(D_tP_t\right) e^{rt}}{P_0} - 1 \right]a^\star.
\end{align*}
 We demonstrate in the following section that a wisely chosen stopping time can yield a positive profit.

\subsection{No Transaction Fees}
\label{sc: without transaction fees}

We examine the investor's payoff in the absence of transaction fees, given by,
\begin{equation*}
    M\left(t, P_t \right) =   e^{-\rho t} \,  \left( - \frac{e^{rt} D_t P_t}{2 P_0} + \sqrt{\frac{D_t P_t}{P_0}} - \frac{ D_t P_t}{2 P_0} \right). 
\end{equation*}
Therefore, our objective function can be expressed as
\begin{align*}
    \mathbb{E} \left[ M \left( T_{c,d}, P_{T_{c,d}} \right) \right] =\begin{cases}
    T_4^+ + T_5^+ + T_6^+, &\text{ when } c > 0, \\
     T_4^- + T_5^- + T_6^- ,  &\text{ when } c < 0.
\end{cases}
\end{align*}

By analyzing the terms $T_1^+, T_1^-,\dots, T_6^+,T_6^-$, we establish the following result.
\begin{proposition}
The following results hold:
\begin{enumerate}
    \item $\mathbb{E} \left[ M \left(T_{c,d}, P_{T_{c,d}} \right) \right] > 0  \ \text{ when } c < 0$; 
    \item $\mathbb{E} \left[ M \left(T_{c,d}, P_{T_{c,d}} \right) \right] < 0 \ \text{ when } c >0$.
\end{enumerate}
\label{prop: no fees obj value pos neg}
\end{proposition}
The proof of Propostion \ref{prop: no fees obj value pos neg} is provided in Appendix \ref{appx: proof of results in sc 4}.

\begin{remark}
    The proposition demonstrates that selecting a stopping level $c < 0$, which corresponds to $L < P_0$, can lead to a positive profit, despite the fact that
    \begin{equation*}
        -e^{-\rho t}x^\star \mathbb{E} \left[ \left( \sqrt{D_tP_t} - \sqrt{P_0} \right)^2 \right] - e^{-\rho t} \,\left[\frac{\mathbb{E}\left(D_tP_t\right) e^{rt}}{P_0} - 1 \right]a^\star < 0,
    \end{equation*}
    as discussed at the end of Section \ref{sc: optimal exit problem formulation}.
\end{remark}

\quad At first glance, it may appear counterintuitive that a positive profit can happen when $c < 0$. Under standard expectations, an investor’s total profit and loss cannot be positive without earning transaction fees. However, in an optimal stopping framework, the stopping time is path-dependent. This allows the investor to selectively stop along trajectories where the total profit and loss happens to be positive. Proposition~\ref{prop: no fees obj value pos neg} formally demonstrates that it is indeed possible to achieve a positive total profit and loss for certain values of $c < 0$. To further understand the origin of this phenomenon, we decompose the investor’s net profit into two components, as presented in the following proposition.

\begin{proposition}
In the absence of transaction fees, the investor's net profit at time 
$t$, relative to pure staking, can be decomposed as follows:
\begin{equation*}
    M \left(t, P_t \right) = ST\left(t,P_t \right) + IL\left(t,P_t \right), 
\end{equation*}
where 
\begin{equation}
\begin{aligned}
    ST \left( t, P_t\right) &= -\frac{1}{2} e^{-\rho t} \left ( \frac{D_t P_t e^{rt}}{P_0} - 1\right), \\
    IL \left( t, P_t\right) &= -\frac{1}{2P_0} e^{-\rho t} \left( D_tP_t + P_0 - 2\sqrt{P_0D_tP_t}  \right),\\
    & = -\frac{1}{2P_0} e^{-\rho t} \left( \sqrt{D_tP_t} - \sqrt{P_0} \right)^2.
    \label{eq: ST and IL}
\end{aligned}
\end{equation}
Here, $IL\left(t,P_t\right)$ represents the impermanent loss arising from liquidity provision in the AMM pool, and $ST \left( t, P_t \right
)$ captures the opportunity cost due to staking only half an ETH in the liquid staking pool (LSP), instead of the full amount of one ETH.
\label{prop:decomp ST and IL}
\end{proposition}

\quad In the following two propositions, we analyze the behavior of the two components, $ST\left(t,P_t \right)$ and $IL\left(t,P_t \right)$, respectively.

\begin{proposition} \leavevmode
\begin{enumerate}
    \item $\mathbb{E} \left[ ST \left(T_{c,d}, P_{T_{c,d}} \right
    ) \right] > 0$ when $c < 0$;
    \item $\mathbb{E} \left[ ST \left(T_{c,d}, P_{T_{c,d}} \right
    ) \right] < 0$ when $c > 0$.
\end{enumerate}
\label{prop: obj staking reward behavior}
\end{proposition}

The proof of Proposition \ref{prop: obj staking reward behavior} is given in Appendix \ref{appx: proof of results in sc 4}.

\begin{proposition}
The following result always holds:
    \begin{equation*}
        \mathbb{E} \left[IL\left( T_{c,d}, P_{T_{c,d}}\right)\right] < 0,
    \end{equation*}
when $c > 0$ or $c < 0$.
\label{prop: obj impermanent loss behavior}
\end{proposition}

The proof of Proposition \ref{prop: obj impermanent loss behavior} is provided in Appendix \ref{appx: proof of results in sc 4}.

\begin{remark}
    These two propositions yield several important insights. First, the impermanent loss is always negative, regardless of the investor's choice of stopping levels. This aligns with the fact that
\begin{equation*}
    -e^{-\rho t}x^\star \mathbb{E} \left[ \left( \sqrt{D_tP_t} - \sqrt{P_0} \right)^2 \right] < 0,
\end{equation*}
as described in Section \ref{sc: optimal allocation}. 

\quad Second, the opportunity cost becomes positive when $c < 0$. This is because, when the price declines, the portion of ETH that remains unstaked (i.e., not staked through the LSP) is not affected by the decrease in the value of LSTs, thereby compensating for the loss incurred from staking via the LSP and holding the LSTs. Consequently, the so-called opportunity cost is not a true ``cost'' in this case, but rather a gain. Moreover, Proposition \ref{prop: no fees obj value pos neg} implies that,
\begin{equation*}
    \left| \mathbb{E} \left[ ST\left( T_{c,d}, P_{T_{c,d}}\right) \right] \right| > \left| \mathbb{E} \left[ IL\left( T_{c,d}, P_{T_{c,d}}\right) \right] \right|.
\end{equation*}
That is, the profit from allocating $\frac{1}{2}$ ETH to stake through LSP and $\frac{1}{2}$ ETH to hold is not only positive, but is also sufficient to outweight the impermanent loss.

\quad Third, when $c > 0$, both $\mathbb{E} \left[ ST\left( T_{c,d}, P_{T_{c,d}}\right) \right]$ and $ \mathbb{E} \left[IL\left( T_{c,d}, P_{T_{c,d}}\right) \right]$ are negative, and therefore $\mathbb{E} \left[M \left(T_{c,d}, P_{T_{c,d}} \right) \right] < 0$. In this case, the price increase leads to a genuine opportunity cost, since the investor could have staked the full amount of 1 ETH through the LSP. By doing so, she would have earned more staking rewards and also benefited from the increase in the value of LSTs, which is not realized by merely holding ETH.
\end{remark}

\quad As discussed above, a positive reward is achievable only through the allocation strategy, which embodies the principle of risk allocation described in Remark \ref{rmk: opt allocation risk}.

\quad According to Proposition \ref{prop: no fees obj value pos neg}, the maximum of $\mathbb{E} \left[ M \left( T_{c,d}, P_{T_{c,d}}\right) \right]$ must occur at some $c < 0$. While the closed-form expression of the optimal $c^\star$ is difficult to make explicit in general, it can be characterized under suitable conditions, as given in Proposition \ref{prop: optimal c characterization} in Appendix \ref{appx: proof of results in sc 4}.


\quad We perform numerical experiments to examine the behavior of the objective function $\mathbb{E} \left[ M \left( T_{c,d}, P_{T_{c,d}}\right) \right]$ and its components, $\mathbb{E} \left[IL\left( T_{c,d}, P_{T_{c,d}}\right)\right]$ and $\mathbb{E} \left[ST\left( T_{c,d}, P_{T_{c,d}}\right)\right]$. A more detailed analysis of the experimental results is provided in Section~\ref{sc: numerical studies}. These numerical findings corroborate our theoretical results.

\subsection{With Transaction Fees} 
\label{sc: with transaction fees}

Let $N\left(t, P_{t} \right)$ denote the transaction fee upper bounded by $2P_0K$ (as discussed at the end of Section \ref{sc: optimal allocation}), when $x^\star = \frac{1}{2P_0}$ and $a^\star = \frac{1}{2}$. Let $\mathbb{E} \left[N\left(T_{c,d}, P_{T_{c,d}} \right) \right]$ denote the expectation of $N\left(t, P_{t} \right)$ under the stopping time $T_{c,d}$. That is,
\begin{align*}
   \mathbb{E} \left[N\left(T_{c,d}, P_{T_{c,d}} \right) \right]
&=
\begin{cases}
    \min  \left\{T_1^+ + T_2^+ + T_3^+, K \right\} ,  &\text{ when } c > 0, \\
     \min  \left\{T_1^- + T_2^- + T_3^-, K \right\} ,  &\text{ when } c > 0.
\end{cases}
\end{align*}

\begin{lemma}
The following result always holds:
\begin{equation*}
 \mathbb{E} \left[N\left(T_{c,d}, P_{T_{c,d}} \right) \right] > 0,
\end{equation*}
when $c > 0$ or $c < 0$.
\label{prop: fee behavior}
\end{lemma}
\begin{proof}
    It suffices to observe that $ 2\left(g+r-m-\rho\right) > g - \frac{1}{4} \sigma^2 - m - 2\rho$ and $2\left(g-m-\rho\right) >g - \frac{1}{4} \sigma^2 - m - 2\rho$ always hold under Assumption \ref{assump: stake via LSP}.
\end{proof}

\begin{proposition}
    $E\left[ W\left(T_{c,d},P_{T_{c,d}} \right)\right] > 0$ when $c < 0$.
    \label{prop: 4.9}
\end{proposition}

\quad Proposition \ref{prop: fee behavior} indicates that the transaction fee remains positive for both $c < 0$ and $c > 0$. Combining this observation with Proposition \ref{prop: no fees obj value pos neg}, we may conclude that Proposition \ref{prop: 4.9} holds. In fact, with transaction fees, our numerical findings indicate that the optimal solution $c^\star$  occurs at some $c < 0$. While the objective function without fees is positive when $c < 0$ and negative when $c > 0$, it is important to note that we set the transaction fee at its lower bound. (This choice reflects the fact that it is not beneficial for the AMM to impose high fees on traders, as excessive fees may discourage trading activity; accordingly, we also impose an upper bound on the fee.) As established in Theorem~\ref{thm: bounds for fees}, the transaction fee is expressed as a sum of exponential functions and depends on time $t$. However, under the stopping time $T_{c,d} = \inf\left\{t \geq 0: B_t = c + dt \right\}$ with $d > 0$ , the event of hitting a positive threshold $c > 0$  is relatively unlikely, given by $\mathbb{P} \left(T_{c,d}  < \infty \right) = e^{-2cd}$. As a result, the expected transaction fee under $T_{c,d}$ for $c > 0$ remains low and is insufficient to offset the negative objective value in the absence of fees. 

\quad When $c < 0$, the exponential form of the fee allows it to grow, when either the reward rate $r$ or the price growth rate $g$ is large. Given that the objective function without transaction fees is already positive in this region, it is not surprising that the addition of transaction fees preserves this positivity. As shown in the plots in Section~\ref{sc: numerical studies}, the optimal $c^\star$ is more negative in the presence of fees than in their absence. This suggests that transaction fees may encourage investors to remain in the AMM for a longer period. A more detailed presentation and discussion of the numerical results are provided in Section \ref{sc: numerical studies}.

\section{Numerical Studies and Further Insights}
\label{sc: numerical studies}

\quad In this section, we conduct a series of numerical experiments to validate the theoretical predictions and gain further insights into the economic implications of the investor’s liquidity provision and exit strategy. Several key observations emerge from this analysis. Section \ref{sc: analysis of the objective function} examines the decomposition of the objective function, focusing on the respective roles of impermanent loss, opportunity cost, and transaction fee in shaping investor incentives. Section \ref{sc: implications for fee structure} explores how the transaction fee structure affects the objective function and investor's participation decision. Section \ref{sc: characterization of the optimal exit threshold} characterizes the optimal exit threshold. Section~\ref{sc: parametric analysis of optimal exit} analyzes how the presence of transaction fees and variations in model parameters affect the optimal exit strategy. Finally, Section \ref{sc: free boundary} presents the numerical results for the free-boundary problem defined in \ref{prob: free boundary}.

\subsection{Analysis of the Objective Function}
\label{sc: analysis of the objective function}
The objective function can be decomposed into the sum of three components: the opportunity cost of capital, impermanent loss, and accrued transaction fee income. The explicit expressions for each component are provided in Appendix~\ref{sc: decomposition}. We conduct numerical experiments to understand how each component contributes to the overall objective. In this analysis, we separately consider the cases where $c$ is positive and negative.

\begin{enumerate}
    \item For $c>0$, we consider the stopping time $T = \inf\left\{t \geq 0: P_t \geq L\right\}$, which corresponds to a stop-win strategy.

    \item For $c<0$, we consider the stopping time $T = \inf\left\{t \geq 0: P_t \leq L\right\}$, which corresponds to a stop-loss strategy. 
    
\end{enumerate}

\begin{figure}[h]
    \centering
    \includegraphics[width=\linewidth]{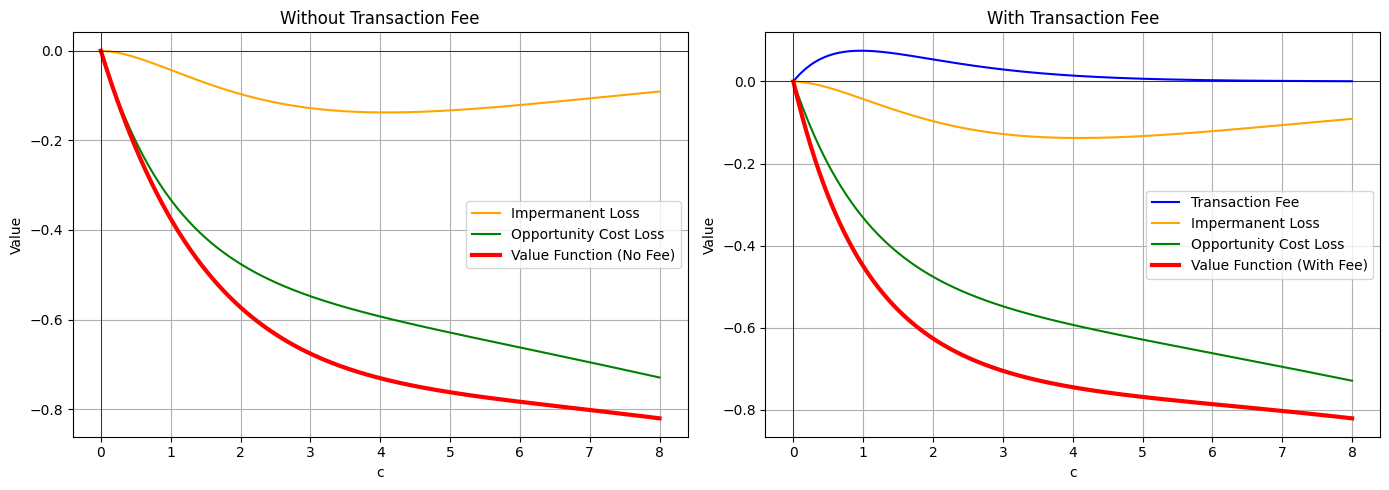}
    \caption{Objective function when $c>0$}
    \label{fig: objective function with postive c}
\end{figure}

\begin{figure}[h]
    \centering
    \includegraphics[width=\linewidth]{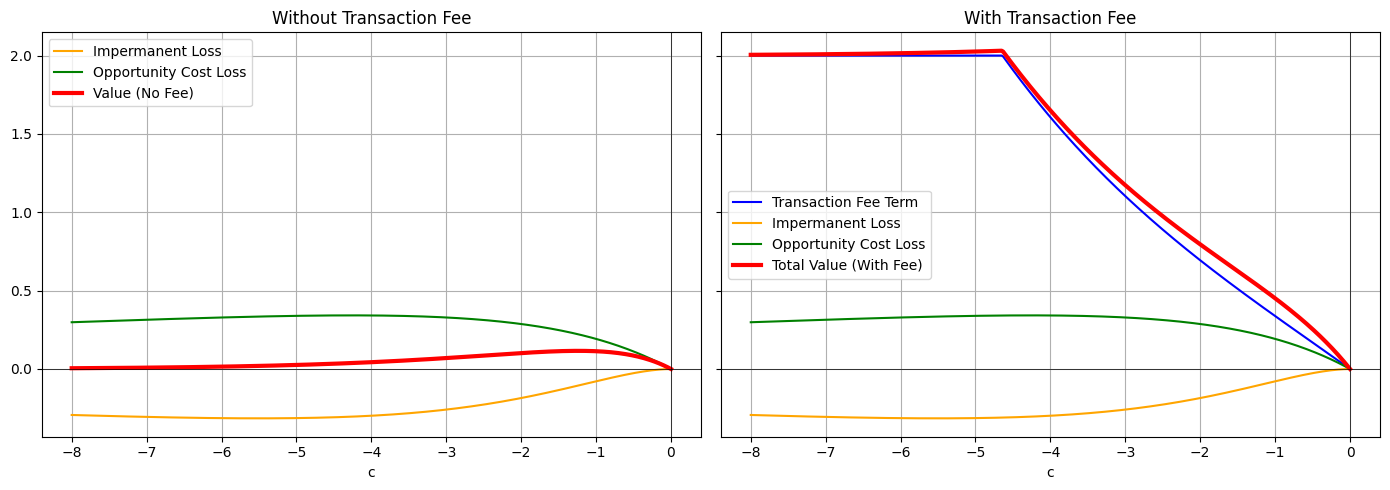}
    \caption{Objective function when $c<0$}
    \label{fig: objective function with negative c}
\end{figure}

\quad Figure~\ref{fig: objective function with postive c} and ~\ref{fig: objective function with negative c} plot the three components: impermanent loss, opportunity cost loss, and transaction fee, as well as the overall objective function, both with and without transaction fees. 

\quad From Figure~\ref{fig: objective function with postive c}, we observe that when $c > 0$, the investor suffers negative impermanent loss and opportunity cost from not staking fully. Both terms are negative, and their combined magnitude outweighs the transaction fee income. As a result, the overall value function remains negative, regardless of whether transaction fees are present. This implies that for an investor in the stop-win strategy, liquidity provision is strictly worse than staking.

\quad From Figure~\ref{fig: objective function with negative c}, we observe that when $c < 0$, the investor benefits from holding ETH instead of staking, since the value of unstaked ETH remains unaffected by the price decline. This leads to a positive opportunity gain, in contrast to the loss observed under the stop-win strategy.

\quad Interestingly, we find that the opportunity gain can, in some cases, outweigh the impermanent loss, resulting in a positive total value even in the absence of transaction fees. While this may seem counterintuitive, the intuition lies in the directional nature of the price drop. Unlike impermanent loss, which arises from relative price movement between the two tokens in a liquidity pool, the opportunity gain here comes from a favorable directional exposure to ETH price: by avoiding staking, the investor avoids locking in the losses from a falling ETH price. As a result, the net value from not staking can surpass the loss induced by impermanent loss in the LP position.

\subsection{Implications for Fee Structure}
\label{sc: implications for fee structure}

\quad As observed in Figure~\ref{fig: objective function with negative c}, transaction fees accumulate at a significantly faster rate and to a much higher level when $c < 0$, compared to the $c > 0$ regime. This asymmetry plays a critical role in shaping the optimal strategy discussed in Section~\ref{sc: characterization of the optimal exit threshold}. Here, we provide an explanation for the observed shape of the fee term.

\quad Mathematically, the fee expression consists of exponential terms of the form $e^{-Dc}$, where $D = d \pm \sqrt{\cdot}$, and the entire term is subject to a cap $K$. Since the constants $D$ are strictly positive under typical parameterizations, the exponential expression grows large when $c < 0$ and shrinks rapidly toward zero when $c > 0$. This explains why the fee component dominates in the left tail of the objective function, especially when $g$ or $r$ is relatively large.

\quad In addition, consider the stopping time $T_{c,d} = \inf\{t \geq 0: B_t = c + dt\}$ with $d > 0$. Then the probability that $T_{c,d}$ is finite depends on the sign of $c$:
\[
\mathbb{P}(T_{c,d} < \infty) = 
\begin{cases}
e^{-2cd}, & c > 0 \\
1, & c < 0
\end{cases}
\]
This means that for $c > 0$, the process is unlikely to hit the target threshold as $c$ increases, while for $c < 0$, the process almost surely reaches the threshold. In a positively drifting setting, this implies that downward barriers (i.e., lower bound $L$) are much more likely to be hit than upward ones.

\quad Together, these effects create a compelling asymmetry: fees accumulate rapidly when $c < 0$ both due to the exponential structure of the fee formula and because hitting the stopping boundary is more probable. This asymmetry, combined with the fee cap $K$ that limits how long an investor can remain in the AMM pool, provides additional intuition as to why the optimal exit is more likely to be characterized by a lower bound $L^* < P_0$ (i.e., a stop-loss strategy), rather than an upper bound $L^* > P_0$. This point will be further explored in Section~\ref{sc: characterization of the optimal exit threshold}.

\subsection{Characterization of the Optimal Exit Threshold}
\label{sc: characterization of the optimal exit threshold}
To determine the optimal stopping threshold $L$ (equivalently, the optimal final log price change $c$), we analyze the shape of the objective function over c.

\begin{figure}[h]
    \centering
    \includegraphics[width=\linewidth]{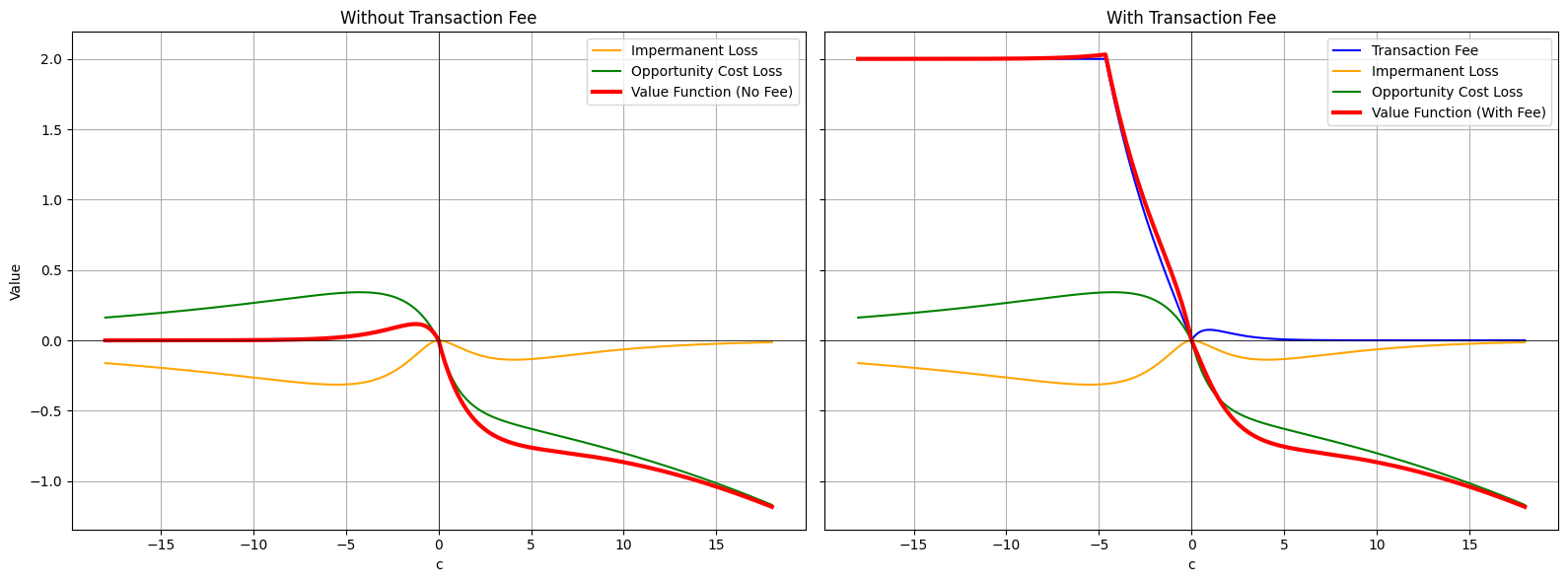}
    \caption{Shape of objective function }
    \label{fig: shape of objective function}
\end{figure}

\quad Our objective is to maximize the value function of $c$. As illustrated in Figure~\ref{fig: shape of objective function}, and consistent with the component-wise analysis in Section~\ref{sc: analysis of the objective function}, the objective function remains strictly negative for all $c > 0$. This implies that under a stop-win strategy, the investor always incurs losses from liquidity provision, making such a strategy not optimal. Since the value function remains negative in this region, it suggests that an upper bound on the stopping threshold is unlikely to be probable in practice.

\quad In contrast, for $c<0$, the objective function exhibits a concave shape with a unique maximum at an interior point $c^*<0$. To the left of $c^*$, the value function decreases gradually; to the right, it drops sharply and converges to zero as $c \to 0^{-}$. This structure implies the existence of an optimal negative threshold $c^*$ that maximizes the investor’s expected value, and hence determines the optimal lower bound $L^* = P_0 e^{\sigma c^*}$ for price-based exit.

\quad In this region, the optimal strategy is to wait until the asset price falls to a certain low level before exiting the pool, which is regarded as a stop-loss strategy. Interestingly, as shown in Figure~\ref{fig: shape of objective function} and discussed in Section~\ref{sc: analysis of the objective function}, it is not really loss if we have transaction fee.

\quad The economic rationale behind this strategy is as follows: although a declining price leads to temporary losses, the longer the price remains depressed, the more fees the investor accrues. In particular, the fee generation is highly convex in $c$ when $c < 0$, i.e., the speed and amount of fee accumulation increase sharply as price drops (which has been discussed in ~\ref{sc: implications for fee structure}). This makes it worthwhile for the investor to endure the price downturn, as the anticipated fee revenue can outweigh the losses and justify remaining in the liquidity pool. The optimal exit therefore reflects a deliberate trade-off: the investor accepts temporary losses from price declines in exchange for rapidly accumulating transaction fee gain.


\subsection{Parametric Analysis of the Optimal Exit Strategy}
\label{sc: parametric analysis of optimal exit}
We conduct a parametric analysis to examine how the investor's optimal exit threshold responds to changes in different settings and parameters. In this analysis, we consider two cases:

\begin{enumerate}
    \item We compare the optimal exit threshold with and without transaction fees.
    \item We examine how the optimal exit threshold responds to changes in model parameters.
\end{enumerate}

\begin{table}[h] 
\centering
\begin{tabular}{ |c|c|c|c|c|c|c| } 
 \hline
 \rule{0pt}{10pt} & Staking Reward $r$ & 0.12 & 0.14 & 0.16 & 0.18 & 0.20 \\ 
 \hline 
 \multirow{3}{*}{No Fees} 
  & Scaled Threshold $c^\star$ & -1.291 & -1.298 & -1.301 & -1.292 & -1.236 \\ 
 & Optimal Exit Price $L^\star$ &0.315 &0.313 &0.312 &0.315 & 0.331 \\
 & Optimal Value $V^\star$ &0.187 &0.176 & 0.163 &0.145 & 0.116 \\
 \hline
 \multirow{3}{*}{With Fees} 
 & Scaled Threshold $c^\star$ & -63.688 & -18.931 & -10.538 & -6.885 & -4.640 \\ 
 &Optimal Exit Price $L^\star$ &1.82e-25 & 4.43e-8 & 8.07e-5 & 2.12e-3 & 0.016 \\
 &Optimal Value $V^\star$ &2.000 &2.000 &2.002 &2.014 & 2.031 \\
 \hline
\end{tabular}

\medskip
\caption{This table presents a parametric analysis of the optimal exit scaled threshold $c^*$ and optimal exit price $L^*$ with and without transaction fees, under varying staking reward $r$, fixed $\rho=0.03$, $\sigma = \sqrt{0.8}$, $g=0$, $K=2$ and $m=0.08$.}
\label{table: exit-threshold-vs-r}
\end{table}

\begin{table}[h] 
\centering
\begin{tabular}{ |c|c|c|c|c|c|c| } 
 \hline
 \rule{0pt}{10pt} & Price Growth $g$ & 0.120 & 0.125 & 0.130 & 0.135 & 0.140 \\ 
 \hline 
 \multirow{3}{*}{No Fees} 
  & Scaled Threshold $c^\star$ & -1.167 & -1.163 & -1.160 & -1.157 & -1.154 \\ 
 & Optimal Exit Price $L^\star$ &0.352 &0.353 &0.354 &0.355 & 0.356 \\
 & Optimal Value $V^\star$ &0.221 & 0.221 & 0.220 & 0.220 & 0.219 \\
 \hline
 \multirow{3}{*}{With Fees} 
 & Scaled Threshold $c^\star$ & -20.613 & -13.311 & -9.733 & -7.546 & -5.990 \\ 
 & Optimal Exit Price $L^\star$ &9.84e-9 & 6.75e-6 & 1.66e-4 & 1.17e-3 & 4.71e-3\\
 & Optimal Value $V^\star$ &2.000 &2.000 &2.002 &2.009 & 2.022 \\
 \hline
\end{tabular}

\medskip
\caption{This table presents a parametric analysis of the optimal exit scaled threshold $c^*$ and optimal exit price $L^*$ with and without transaction fees, under varying price growth $g$, fixed $\rho=0.03$, $\sigma = \sqrt{0.8}$, $r=0$, $K=2$ and $m=0.08$.}
\label{table: exit-threshold-vs-g}
\end{table}

\quad The results are presented in Table ~\ref{table: exit-threshold-vs-r} and Table ~\ref{table: exit-threshold-vs-g}. Additional results are shown in Table ~\ref{table: exit-threshold-vs-K} to Table ~\ref{table: exit-threshold-vs-rho2} in Appendix~\ref{appx: additional numerics}. By comparing the values in each column across the ``with fees" and ``no fees" rows of a given table, we can observe how the presence of transaction fees alters the investor's optimal exit threshold. Furthermore, by analyzing the "with fees" rows across different columns within the same table, we can analyze how changes in individual parameters (such as staking reward rate $r$ in Table ~\ref{table: exit-threshold-vs-r} or price growth rate $g$ in Table ~\ref{table: exit-threshold-vs-g}) affect the investor's optimal exit behavior.

\quad We observe that, in most cases, the optimal exit threshold is substantially lower when transaction fees are present, suggesting that investors remain in the pool for a longer period. This observation aligns with the analysis in Section~\ref{sc: characterization of the optimal exit threshold}, which shows that fee accumulation grows rapidly as the price declines. Hence, in the presence of fees, the investor finds it beneficial to remain in the AMM pool longer to accrue more fees, even at the cost of impermanent loss.

\quad Furthermore, in the presence of transaction fees, Table ~\ref{table: exit-threshold-vs-r}, Table ~\ref{table: exit-threshold-vs-g}, and Table ~\ref{table: exit-threshold-vs-K} to Table ~\ref{table: exit-threshold-vs-rho2} reveal a clear and consistent pattern: the optimal exit threshold decreases as the term $r + g - \rho - m$ in Assumption ~\ref{assump: stake via LSP} becomes smaller, indicating delayed exit from the liquidity provision. The term $r + g - \rho - m$ in Assumption ~\ref{assump: stake via LSP} means the difference between the reward rates and the discount rates. It will decrease when either the staking reward $r$ or the price growth rate $g$ decrease, or when the discounting parameters $\rho$ or $m$ increase. As previously discussed in Section~\ref{sc: analysis of the objective function}, the structure of the objective function reveals that, as the difference between the reward rates and the discount rates $r + g - \rho - m$ decreases, the staking reward becomes less attractive relative to fee accumulation, and the objective function becomes increasingly dominated by transaction fee accumulation. This can also be seen in the analysis in Appendix~\ref{appx: additional numerics}. Intuitively, this phenomenon
encourages the investor to delay exit in order to harvest more transaction fees, which further validates the stop-loss strategy identified in Section ~\ref{sc: characterization of the optimal exit threshold}.

\subsection{Numerical Analysis of the Free-Boundary Problem}
\label{sc: free boundary}
We perform numerical experiments to solve the free-boundary problem defined in \ref{prob: free boundary}.

\begin{figure}[h]
    \centering
        \includegraphics[width=0.8\linewidth]{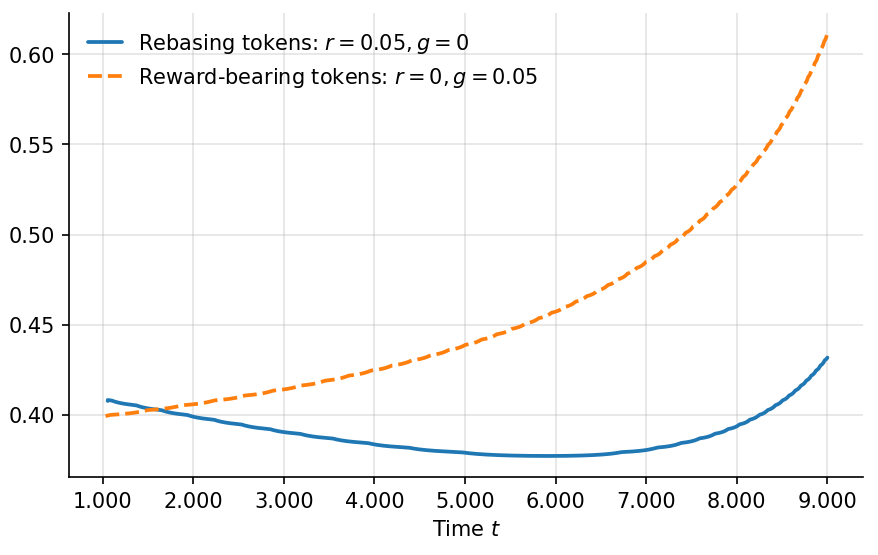}
    \caption{Evolution of the boundary for rebasing and reward-bearing tokens.}
    \label{fig:free_boundaries}
\end{figure}

\begin{figure}[h]
    \centering
    \begin{subfigure}[t]{0.32\linewidth}
        \renewcommand\captionlabelfont{}
        \centering
        \includegraphics[width=\linewidth]{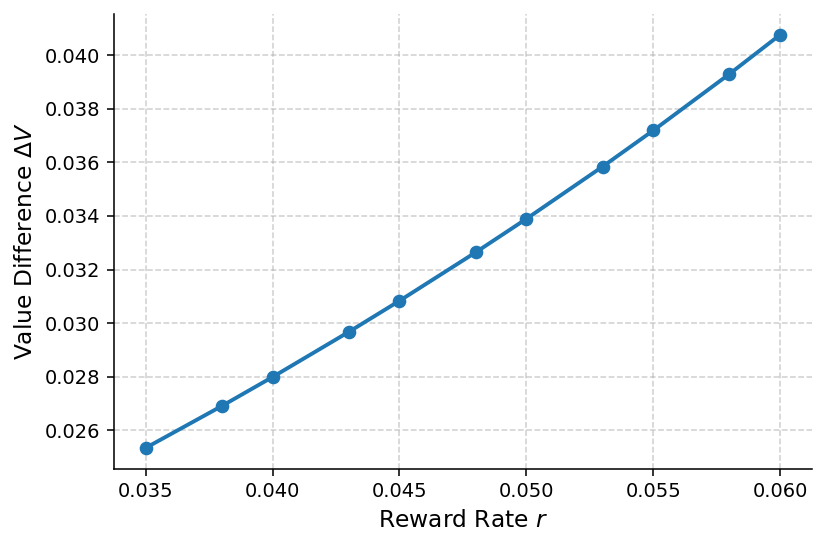}
        \caption{}
        \label{fig:gaps_r}
    \end{subfigure}
    \begin{subfigure}[t]{0.32\linewidth}
        \renewcommand\captionlabelfont{}
        \centering
        \includegraphics[width=\linewidth]{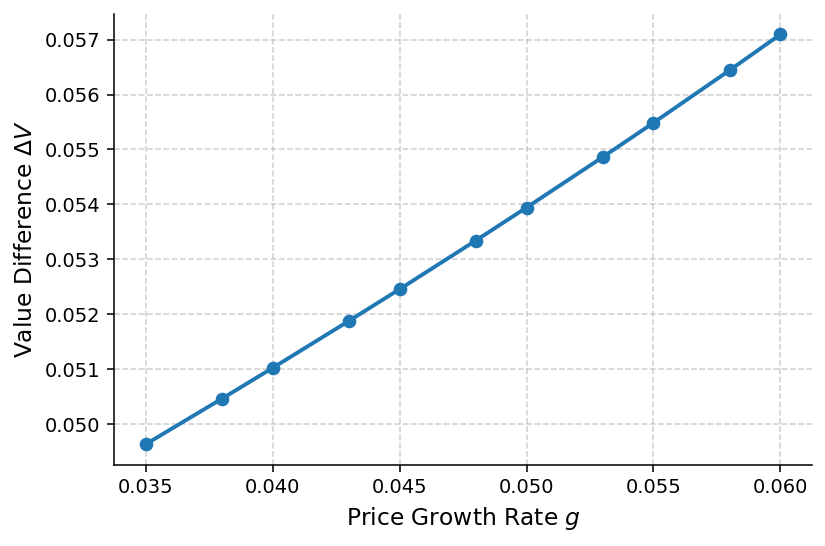}
        \caption{}
        \label{fig:gaps_g}
    \end{subfigure}
    \begin{subfigure}[t]{0.32\linewidth}
        \renewcommand\captionlabelfont{}
        \centering
        \includegraphics[width=\linewidth]{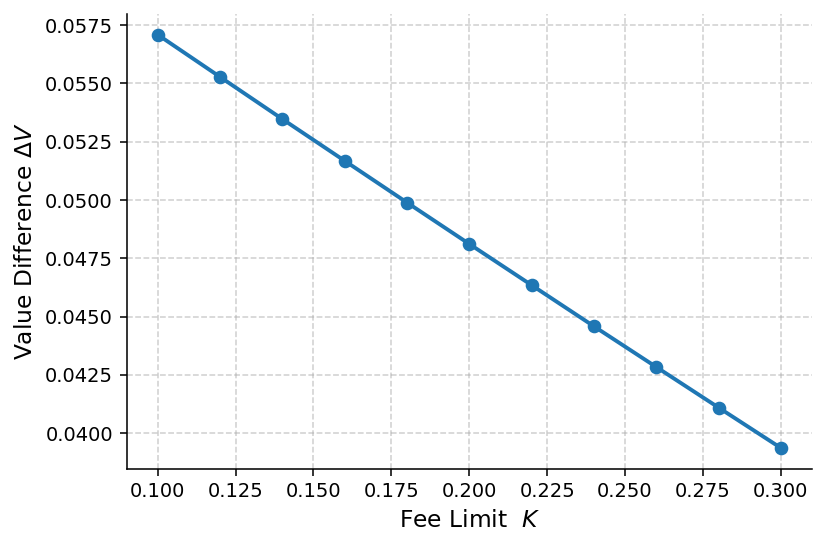}
        \caption{}
        \label{fig:gaps_K}
    \end{subfigure}
    \caption{Changes in optimal value differences with reward rate $r$, growth rate $g$, and fee limit $K$}
    \label{fig:gaps vs g r}
\end{figure}

\quad Figure \ref{fig:free_boundaries} presents the evolution of the free boundary for reward-bearing tokens and rebasing tokens. The parameter configurations for the two cases are as follows: $\sigma = 0.8, \ \rho = 0.02, \ m = 0.01$, and $\ K = 0.1$. For rebasing tokens, the parameters are $g = 0.05$ and $\ r = 0$, while for reward-bearing tokens, they are $g = 0, r = 0.05$. Under the first set of parameters, the optimal value of the free-boundary problem is $0.345$, compared to $0.329$ for the fixed stopping level. Under the second set, the corresponding optimal values are $0.338$ and $0.304$, respectively.

\quad In both plots, the free boundaries begin at some $t > 0$. This occurs because the investor’s total profit and loss comprises two components: transaction fees and the value of the investor’s portfolio, which we later refer to as the ``total payoff'' in the figures in Appendix~\ref{appx sub: free boundary}. It is initially advantageous for the investor to allow transaction fees to accumulate, given that the fee depends solely on $t$. However, as $t$ increases, although the transaction fee may continue to rise until it reaches its upper limit $K$, the portfolio’s optimal value, represented by the second term in Eq.~\eqref{eq: Z(t,x)}, may decline. Consequently, the optimal free boundary must balance these two opposing effects.

\quad As shown in Proposition~\ref{prop:decomp ST and IL}, the investor’s portfolio value can be decomposed into two components: opportunity gain, arising from staking only half of the ETH, and impermanent loss. When considering the trajectory-wise decision-making problem (rather than expectations), impermanent loss is independent of both the reward rate $r$ and the price growth rate $g$, as shown in Eq.~\eqref{eq: ST and IL}. For the opportunity gain term, $ST(t, P_t)$ as defined in Eq.~\eqref{eq: ST and IL}, the second term $\frac{1}{2} e^{-\rho t}$ decreases over time, while the coefficient $e^{(r - m - \rho)t}$ of the first term $-\frac{1}{2} e^{(r - m - \rho)t} x$ increases with $t$ for rebasing tokens ($r > 0$) and decreases with $t$ for reward-bearing tokens ($r = 0$). Therefore, as time $t$ grows, the investor’s portfolio value (i.e., the total payoff) for rebasing tokens declines more sharply as price $x$ increases, as illustrated in Figures~\ref{fig:payoff decomp_r} to \ref{fig:total_comparison} in Appendix~\ref{appx sub: free boundary}. As a result, reward-bearing tokens can tolerate higher prices, whereas rebasing tokens tend to have lower optimal exit prices.

\quad From the optimal values, we observe that the free-boundary stopping strategies yield higher optimal outcomes, as expected. However, the improvement over fixed stopping levels is relatively modest. Figure~\ref{fig:gaps vs g r} illustrates how the differences between the optimal values under the free-boundary strategy and those under the fixed stopping-level strategy vary with parameters $r$, $g$, and $K$, where
\begin{equation*}
\Delta V = V^\star_{\text{free}} - V^\star_{\text{fixed}},
\end{equation*}
and $V^\star_{\text{free}}$ and $V^\star_{\text{fixed}}$ denote the optimal values for the free-boundary and fixed stopping-level strategies, respectively. The value gap $\Delta V$ increases with the reward rate $r$ and the price growth rate $g$, but decreases as the upper bound on transaction fees $K$ increases. This pattern is intuitive: higher values of $r$ or $g$ imply greater profit potential, making a flexible (unconstrained) strategy more advantageous relative to a fixed stopping rule. In contrast, when the transaction fee bound $K$ increases, the transaction fees constitute a more significant component of the investor’s total payoff. In such cases, remaining in the position longer becomes more beneficial and the precise timing of exit plays a less critical role. However, in general, the differences between the optimal values are not substantial, suggesting that a simple fixed stopping-level strategy may be more practical in real-world applications.


\section{Discussion}

\quad In this paper, we follow the settings of \citet{milionis2024}, which do not include traders. A direction worth exploring is to extend the model by explicitly incorporating traders. Traders play a dual role in the interaction among liquidity providers (LPs), traders, and the platform: 1) their trading volume determine the transaction fees earned by LPs, thereby affecting the amount of liquidity deposited into the AMM; 2) they influence the price within the AMM pool, which is eventually aligned with the external market price by arbitrageurs \citep{caodavid2025structuralmodel}. Our analyses focuses on liquid staking, with the AMM being part of the whole system. 
Specifically, we derive lower bounds for cumulative transaction fees that provide incentives for both LPs and traders, and model the price as eventually being aligned with the external market through arbitrageurs. We propose a feasible transaction fee structure that achieves this bound. An extension of our framework would be to explicitly model traders and account for trading volume into the derivation of the bounds, as well as into the design of a fee structure that achieves these bounds. 


\section{Concluding Remarks}
\label{sc: conclusion}

\quad In this paper, we study the optimal allocation problem faced by an investor holding ETH across direct holding, liquidity provision, and liquid staking, as well as the optimal timing to exit liquidity provision. We derive the necessary and sufficient conditions under which the investor chooses to provide liquidity and design a fee mechanism that aligns incentives for both the investor and the platform. We characterize the optimal exit time from the AMM pool and analyze the dynamics of the investor’s payoff. Without transaction fees, we analytically investigate the investor’s payoff dynamics and derive the first-order condition for the optimal stopping threshold. With transaction fees, we supplement our analysis with numerical experiments. In both scenarios, we find that a stop-loss strategy emerges as the most economically rational behavior.

\bigskip
{\bf Acknowledgement}:
We thank Tim Roughgarden for helpful discussions. Ruofei Ma is supported by Columbia-Ethereum Foundation Research Center for Blockchain Protocol Design fellowship, the Columbia Innovation Hub grant, and the Center for Digital Finance and Technologies (CDFT) grant. Wenpin Tang acknowledges financial support by NSF grant DMS-2206038, and the Tang Family Assistant Professorship. The works of Ruofei Ma and David Yao are part of a Columbia-CityU/HK collaborative project that is supported by InnoHK Initiative, The Government of the HKSAR and the AIFT Lab.

\bibliographystyle{abbrvnat}
\bibliography{References}

\newpage

\appendix 

\begin{center}
 \large \textbf{Appendix} 
\end{center}

\section{Proof of Results in Section \ref{sc: optimal allocation}}
\label{appx: proof of results in sc 3}

\textbf{Proof of Lemma \ref{lem:lmda pool}.}

\begin{proof}
    Observe that the objective function $f\left(U,V \right) = D_tP_tU + V$ satisfies 
    \begin{equation}
       f\left(\lambda U, \lambda V \right) = \lambda f\left(U, V\right), \forall \lambda > 0,
       \label{eq: lmda obj func}
    \end{equation}

    and the constraint function $h
    \left(U,V\right) = UV = L$ satisfies
    \begin{equation}
       h\left(\lambda U, \lambda V \right) = \lambda^2 h\left(U,V\right) ,\forall \lambda > 0.
       \label{eq: lmda constr func}
    \end{equation}
    Note that $L^\prime = \lambda^2 L$. 
    
    \quad We first prove that $V\left(P_t; L^\prime \right) = \lambda V\left(P_t; L \right)$. Let $\left(u^\star, v^\star \right)$ be the minimizer of Problem \ref{prob: pool value problem single}, so that $u^\star v^\star = L^\prime = xy$, and $V\left(P_t; L^\prime \right) =f\left(u^\star, v^\star \right) = D_tP_t u^\star + v^\star$. Consider the point $\left(\frac{1}{\lambda}u^\star, \frac{1}{\lambda}v^\star \right)$. 
    Since
    \begin{equation*}
        \frac{1}{\lambda^2}u^\star v^\star = \frac{1}{\lambda^2} L^\prime = L,
    \end{equation*}
        $\left(\frac{1}{\lambda}u^\star, \frac{1}{\lambda}v^\star \right)$ is a feasible solution to Problem \ref{prob: pool value problem all}. Thus,
        \begin{equation}
            V\left(P_t;L \right) \leq f\left( \frac{1}{\lambda}u^\star, \frac{1}{\lambda}v^\star \right) = \frac{1}{\lambda} f\left( u^\star, v^\star \right) = \frac{1}{\lambda} V\left( P_t; L^\prime\right),
            \label{ineq: L <= L'}
        \end{equation}
        by Eqs. (\ref{eq: lmda obj func}) and (\ref{eq: lmda constr func}).
        
        \quad Let $\left(U^\star, V^\star \right)$ be the minimizer of Problem \ref{prob: pool value problem all}, so that $U^\star V^\star = L = XY$, and $V\left(P_t; L \right) = f\left(U^\star, V^\star \right) = D_tP_t U^\star + V^\star$. Consider the point $\left(\lambda U^\star, \lambda V^\star \right)$. 
        Since
        \begin{equation*}
            \lambda^2 U^\star V^\star = \lambda^2 L = L^\prime,
        \end{equation*}
        $\left(\lambda U^\star, \lambda V^\star \right)$ is a feasible point for Problem \ref{prob: pool value problem single}, and thus
        \begin{equation}
            V\left(P_t;L^\prime \right) \leq f\left( \lambda U^\star, \lambda V^\star \right) = \lambda f\left( U^\star, V^\star \right) = \lambda V\left( P_t; L^\prime\right),
            \label{ineq: L' <= L}
        \end{equation}
        by Eqs. (\ref{eq: lmda obj func}) and (\ref{eq: lmda constr func}). Combining Inequalities (\ref{ineq: L <= L'}) and (\ref{ineq: L' <= L}), we conclude that $V\left(P_t; L^\prime \right) = \lambda V\left(P_t; L \right)$. Since $\left(\frac{1}{\lambda} u^\star, \frac{1}{\lambda} v^\star \right)$ and $\left(\lambda U^\star, \lambda V^\star \right)$ are feasible for Problems \ref{prob: pool value problem all}  and \ref{prob: pool value problem single}, respectively, and both achieve the optimal value, the results in (\ref{eq: opt sols and value}) follow.
\end{proof}

\smallskip
\textbf{Proof of Proposition \ref{prop: optimal allocation solution} and Theorem \ref{thm: bounds for fees}.} To establish Proposition \ref{prop: optimal allocation solution} and Theorem \ref{thm: bounds for fees}, we first present several supporting lemmas.

\begin{lemma}
    The coefficient of $a$,
\begin{equation}
\frac{\mathbb{E}\left(D_tP_t\right) e^{rt}}{P_0} - 1 =  e^{\left(r+g-m\right)t}  - 1\label{eq: coeff for a},
\end{equation}
is positive. Therefore, we have $a^\star = P_0 x^\star$. 
\end{lemma}

\begin{proof}
Assumption \ref{assump: stake via LSP} ensures that Eq.~\eqref{eq: coeff for a} is strictly positive.
Since (\ref{eq: coeff for a}) $> 0$, we have $a^\star = P_0 x^\star$. 
\end{proof}

\begin{lemma}
    A sufficient condition for $x^\star$ and $a^\star$ to be positive is:
    \begin{equation}
        \int_{0}^t \phi_s e^{-\rho s} ds >  P_0 \left[ e^{\left(g-m-\rho \right)t} \left( e^{rt}+1 \right) - 2  e^{\left(\frac{1}{2}g - \frac{1}{8} \sigma^2 - \frac{1}{2}m - \rho \right)t}  \right], \label{ineq: suff cond}
    \end{equation}
    in which case $x^\star = \frac{1}{2P_0}$ and $a^\star = \frac{1}{2}$.
    \label{lem: suff cond}
\end{lemma}

\begin{proof}

There are are three possible cases, as detailed below.

\begin{enumerate}
    \item If 
    \begin{align*}
        \int_{0}^t \phi_s e^{-\rho s} ds &> e^{-\rho t} \left\{ \mathbb{E} \left[ \left( \sqrt{D_tP_t} - \sqrt{P_0} \right)^2 \right] + \mathbb{E}\left(D_tP_t\right) e^{rt} - P_0 \right\},\\
        &= P_0 \left[ e^{\left(g-m-\rho \right)t} \left( e^{rt}+1 \right) - 2  e^{\left(\frac{1}{2}g - \frac{1}{8} \sigma^2 - \frac{1}{2}m - \rho \right)t}  \right],
    \end{align*}
    $x^\star = \frac{1}{2P_0}$ and $a^\star = \frac{1}{2}$.
    \item If 
    \begin{equation*}
        \int_{0}^t \phi_s e^{-\rho s} ds < P_0 \left[ e^{\left(g-m-\rho \right)t} \left( e^{rt}+1 \right) - 2  e^{\left(\frac{1}{2}g - \frac{1}{8} \sigma^2 - \frac{1}{2}m - \rho \right)t}  \right],
    \end{equation*}
    $x^\star = 0$ and $a^\star = 0$.
    \item If \begin{equation*}
        \int_{0}^t \phi_s e^{-\rho s} ds = P_0 \left[ e^{\left(g-m-\rho \right)t} \left( e^{rt}+1 \right) - 2  e^{\left(\frac{1}{2}g - \frac{1}{8} \sigma^2 - \frac{1}{2}m - \rho \right)t}  \right],
    \end{equation*}
    then any $x^\star \in \left[0,\frac{1}{2P_0} \right]$ and $a^\star = P_0 x^\star$ are optimal. 
\end{enumerate}
Thus, the result follows under Assumption \ref{assump: tie-breaking}.
    
\end{proof}

\begin{lemma}
    Inequality (\ref{ineq: suff cond}) is also a necessary condition for $x^\star$ and $a^\star$ to be positive under  Assumption \ref{assump: tie-breaking}.
    \label{lem: nec condition}
\end{lemma}
\begin{proof}
Since (\ref{eq: coeff for a}) $> 0$, $a^\star > 0$ and $x^\star > 0$ implies \begin{align*}
        \int_{0}^t \phi_s e^{-\rho s} ds 
        > P_0 \left[ e^{\left(g-m-\rho \right)t} \left( e^{rt}+1 \right) - 2  e^{\left(\frac{1}{2}g - \frac{1}{8} \sigma^2 - \frac{1}{2}m - \rho \right)t}  \right],
    \end{align*}
under Assumption \ref{assump: tie-breaking}. Note that under Assumption \ref{assump: tie-breaking}, $a^\star > 0$ and $x^\star > 0$ implies that a positive solution is strictly better than $a^\star = x^\star = 0$, and thus the inequality must be strict.
\end{proof}

\quad Proposition \ref{prop: optimal allocation solution} can be concluded from the proof of Lemma \ref{lem: suff cond}. Combining Lemmas \ref{lem: suff cond} and \ref{lem: nec condition}, Theorem \ref{thm: bounds for fees} follows immediately.

\newpage
\section{Explicit Expressions for \texorpdfstring{$T_i^+$ and $T_i^-$}{Ti+ and Ti-}}
\label{appx: Ti}

The explicit expressions for $T_i^+$ are given as follows.
\begin{align*}
    T_1^+ &= \frac{1}{2} \exp \left( -c\left( d + \sqrt{d^2 - 2\left(g+r-m-\rho\right)}\right)\right), \\
    T_2^+ &=\frac{1}{2} \exp \left( -c\left( d + \sqrt{d^2 - 2\left(g-m-\rho\right)}\right)\right),
    \\
    T_3^+ &= - \exp \left( -c\left( d + \sqrt{d^2 - \left(g - \frac{1}{4} \sigma^2 - m - 2\rho\right)}\right)\right),\\
    T_4^+ &= \exp \left( -c\left( -\frac{1}{2} \sigma+ d + \sqrt{d^2 + m + 2\rho}\right)\right),\\
    T_5^+ &=  -\frac{1}{2} \exp \left( -c\left( - \sigma+ d + \sqrt{d^2 - 2\left(r-\rho-m\right)}\right)\right),\\
    T_6^+ &= - \frac{1}{2}  \exp \left( -c\left( -\sigma+ d + \sqrt{d^2 + 2\left(\rho+m\right)}\right)\right).
\end{align*}

The terms $T_i^-$ are defined similarly from Eq. (\ref{eq: obj value for OST c < 0}), with the plus sign between the term $d$ and the square root term replaced by a minus sign.

\newpage
\section{Proof of Results in Section \ref{sc: optimal time to exit}}
\label{appx: proof of results in sc 4}

\textbf{Proof of Proposition \ref{prop: no fees obj value pos neg}.}

To prove Proposition \ref{prop: no fees obj value pos neg}, we need the following lemma.

\begin{lemma}
    The following results hold:
    \begin{enumerate}
        \item $\frac{1}{2} T_4^- + T_5^- > 0$;
        \vspace{0.5em}
        \item $\frac{1}{2} T_4^- + T_6^- > 0$;
        \vspace{0.5em}
        \item $\frac{1}{2} T_4^+ + T_5^+ < 0$;
        \vspace{0.5em}
        \item $\frac{1}{2}T_4^+ + T_6^+ < 0$.
    \end{enumerate}
    \label{lem: obj value without fees}
\end{lemma}

\begin{proof} \leavevmode
\begin{enumerate}
    \item It suffices to observe that, by Assumption \ref{assump: liquidity provision condition},
    \begin{equation*}
        \sqrt{d^2 + m + 2\rho} - \sqrt{d^2 - 2r + 2\rho + 2m} < 0 < \frac{1}{2} \sigma.
    \end{equation*}

\item It suffices to note that, since $m > 0$,
\begin{equation*}
    \sqrt{d^2 + m + 2\rho} - \sqrt{d^2 + 2\rho+2m} < 0 <\frac{1}{2} \sigma.
\end{equation*}
\item Since $2r - m > 0$ by Assumption \ref{assump: stake via LSP},
it can be easily shown that
\begin{equation*}
    \frac{1}{2} \sigma + \sqrt{d^2+m+2\rho} > \sqrt{d^2 - 2r+2\rho + 2m}.
\end{equation*}
The result then follows immediately.
\item By assumption \ref{assump: range for paras}, 
\begin{equation*}
    \frac{1}{4} \sigma^2 + \sigma\sqrt{d^2+m+2\rho} - m > 0.
\end{equation*}
The result then follows immediately.
\end{enumerate}
\end{proof}

By Lemma \ref{lem: obj value without fees}, Proposition \ref{prop: no fees obj value pos neg} follows immediately.

\bigskip
\textbf{Proof of Proposition \ref{prop: obj staking reward behavior}.}
\begin{proof} \leavevmode
\begin{enumerate}
    \item When $c < 0$,
    \begin{equation*}
        \mathbb{E} \left[ ST\left(T_{c,d}, P_{T_{c,d}} \right
        ) \right] = \frac{1}{2} \left(-e^{-c \left(-\sigma + d - \sqrt{d^2 -2r + 2m + 2\rho} \right)} + e^{-c \left( d - \sqrt{d^2 + 2\rho}\right)} \right).
    \end{equation*}
    Since $m < r$ by Assumption \ref{assump: stake via LSP}, we have
    \begin{equation*}
        \sqrt{d^2 + 2\rho} - \sqrt{d^2 -2r + 2m +2\rho} < 0 < \sigma.
    \end{equation*}
    Therefore, $\mathbb{E} \left[ ST\left(T_{c,d}, P_{T_{c,d}} \right
        ) \right] > 0$.
\item When $c > 0$,
    \begin{equation*}
        \mathbb{E} \left[ ST\left(T_{c,d}, P_{T_{c,d}} \right
        ) \right] = \frac{1}{2} \left(-e^{-c \left(-\sigma + d + \sqrt{d^2 -2r + 2m + 2\rho} \right)} + e^{-c \left( d + \sqrt{d^2 + 2\rho}\right)} \right).
    \end{equation*}
    Since $m < r$ by Assumption \ref{assump: stake via LSP}, we have
    \begin{equation*}
        \sigma^2 + 2r - 2m + 2\sigma \sqrt{d^2 + 2\rho} > 0.
    \end{equation*}
    Therefore, $\mathbb{E} \left[ ST\left(T_{c,d}, P_{T_{c,d}} \right
        ) \right] < 0$.
\end{enumerate}
\end{proof}

\textbf{Proof of Proposition \ref{prop: obj impermanent loss behavior}.}
\begin{proof} \leavevmode
    \begin{enumerate}
        \item When $c > 0$, we have
        \begin{align*}
            \mathbb{E} \left[IL\left( T_{c,d}, P_{T_{c,d}}\right)\right] &= \frac{1}{2} \left(-e^{-c\left(-\sigma + d+\sqrt{d^2+2m+2\rho} \right)} - e^{-c\left(d + \sqrt{d^2+2\rho} \right)} + 2e^{-c\left( -\frac{1}{2} \sigma +d +\sqrt{d^2+m+2\rho}\right)} \right).
        \end{align*}
        Define the constants
        \begin{align*}
            C_1 &= -\sigma + d+\sqrt{d^2+2m+2\rho}, \\
            C_2 &= d + \sqrt{d^2+2\rho}, \\
            C_3 &= -\frac{1}{2} \sigma +d +\sqrt{d^2+m+2\rho}.
        \end{align*}
       Since the function $e^{-cx}$ is strictly convex in $x$, by Jensen's Inequality, we have 
       \begin{equation*}
           \frac{1}{2}e^{-c C_1} + \frac{1}{2}e^{-c C_2} > e^{-c\left( \frac{1}{2}C_1 + \frac{1}{2}C_2 \right)}.
       \end{equation*}
        To complete the proof, it suffices to show that
        \begin{equation}
            \frac{1}{2} C_1 + \frac{1}{2} C_2 < C_3, \label{ineq: for prop IL}
        \end{equation}
        which would imply 
        \begin{equation*}
            e^{-c C_3 } < e^{-c\left(\frac{1}{2}C_1 + \frac{1}{2}C_2 \right)}.
        \end{equation*}
        Showing Inequality (\ref{ineq: for prop IL}) is equivalent to proving the following inequality:
        \begin{equation}
            \sqrt{d^2+m+2\rho} > \frac{\sqrt{d^2+2m+2\rho} + \sqrt{d^2+2\rho}}{2}. \label{ineq: for prop IL equiv.}
        \end{equation}
        Observe that
        \begin{equation*}
            d^2+m+2\rho = \frac{1}{2} \left(d^2 + 2m + 2\rho \right) + \frac{1}{2} \left( d^2+2\rho\right)
        \end{equation*}
        Since the function $\sqrt{x}$ is strictly concave on $\left. \left[0,\infty \right. \right)$, it follows from Jensen's Inequality that Inequality (\ref{ineq: for prop IL equiv.}) holds. 
        \item When $c < 0$, define the constants
        \begin{align*}
             C_1^\prime &= -\sigma + d-\sqrt{d^2+2m+2\rho}, \\
            C_2^\prime &= d - \sqrt{d^2+2\rho}, \\
            C_3^\prime &= -\frac{1}{2} \sigma +d -\sqrt{d^2+m+2\rho}.
        \end{align*}
        It can be proved similarly that
        \begin{equation*}
            \frac{1}{2} C_1^\prime + \frac{1}{2} C_2^\prime > C_3^\prime,
        \end{equation*}
        which would imply
        \begin{equation*}
            \frac{1}{2}e^{-cC1^\prime}+\frac{1}{2}e^{-cC2^\prime} > e^{-c\left(\frac{1}{2}C_1^\prime + \frac{1}{2} C_2^\prime \right)} > e^{-cC_3^\prime} .
        \end{equation*}
        This completes the proof.
    \end{enumerate}
\end{proof}

\begin{proposition}
Let $c^\star = \underset{c < 0}{\arg\max} \; \mathbb{E} \left[ M \left( T_{c,d}, P_{T_{c,d}}\right) \right]$. $c^\star$ is the solution to the following equation
\begin{equation}
     -D_1e^{-cD_1}
   +\frac{1}{2}D_2 e^{-cD_2} + \frac{1}{2} D_3  e^{-cD_3} = 0,
   \label{eq: FOC for c no fees}
\end{equation}
where
\begin{align*}
    D_1 &= -\frac{1}{2} \sigma+ d - \sqrt{d^2 + m + 2\rho}, \\
    D_2 &= - \sigma+ d - \sqrt{d^2 - 2\left(r-\rho-m\right)}, \\
    D_3 &= -\sigma+ d - \sqrt{d^2 + 2\left(\rho+m\right)},
\end{align*}
under the conditions that
\begin{equation}
     D_1^2 e^{-cD_1} - \frac{1}{2} D_2^2 e^{-cD_2} - \frac{1}{2} D_3^2 e^{-cD_3} < 0.
     \label{cond: concavity}
\end{equation}
\label{prop: optimal c characterization}
\end{proposition}

\begin{proof}
    Eq. (\ref{eq: FOC for c no fees}) is the first order derivative of $\mathbb{E} \left[ M \left( T_{c,d}, P_{T_{c,d}}\right) \right]$. Inequality (\ref{cond: concavity}) ensures the the function $\mathbb{E} \left[ M \left( T_{c,d}, P_{T_{c,d}}\right) \right]$ is concave. 
\end{proof}

\newpage
\section{Expressions for Impermanent Loss, Opportunity Cost Loss, and Transaction Fee}
\label{sc: decomposition}
\begin{enumerate}
    \item When $c > 0$, we have:
    \begin{align*}
        \text{Impermanent Loss} &= \frac{1}{2}\left(
        -e^{\sigma c - c\left(d+\sqrt{d^{2}+2m+2\rho}\right)} 
        -e^{-c\left(d+\sqrt{d^{2}+2\rho}\right)} \right. \left.+2e^{\frac{1}{2} \sigma c - c\left(d+\sqrt{d^{2}+m+2\rho}\right)}
        \right),\\
        \text{Opportunity Cost Loss} &= \frac{1}{2}\left(
        -e^{\sigma c - c\left(d+\sqrt{d^{2}-2r+2m+2\rho}\right)} 
        +e^{-c\left(d+\sqrt{d^{2}+2\rho}\right)}
        \right),\\
        \text{Transaction Fee} &= \min \left\{\frac{1}{2} e^{-c\left( d - \sqrt{d^2 - 2\left(g+r-m-\rho\right)}\right)} + \frac{1}{2} e^{-c\left( d - \sqrt{d^2 - 2\left(g-m-\rho\right)}\right)} \right.\\
        & \left. \quad \quad \quad \ \  - e^{-c\left( d - \sqrt{d^2 - \left(g - \frac{1}{4} \sigma^2 - m - 2\rho\right)}\right)}, K \right\}.
    \end{align*}

\item  When $c < 0$, we have
    \begin{align*}
        \text{Impermanent Loss} &= \frac{1}{2}\left(
        -e^{\sigma c - c\left(d-\sqrt{d^{2}+2m+2\rho}\right)} 
        -e^{-c\left(d-\sqrt{d^{2}+2\rho}\right)} \right. \left. +2e^{\frac{1}{2} \sigma c - c\left(d-\sqrt{d^{2}+m+2\rho}\right)}
        \right),\\
        \text{Opportunity Cost Loss} &= \frac{1}{2}\left(
        -e^{\sigma c - c\left(d-\sqrt{d^{2}-2r+2m+2\rho}\right)} 
        +e^{-c\left(d-\sqrt{d^{2}+2\rho}\right)}
        \right), \\
        \text{Transaction Fee} &= \min \left\{
        \frac{1}{2} e^{-c\left( d + \sqrt{d^2 - 2(g + r - m - \rho)}\right)} 
        + \frac{1}{2} e^{-c\left( d + \sqrt{d^2 - 2(g - m - \rho)}\right)} \right.\\
        & \left. \quad \quad \quad \ \ - e^{-c\left( d + \sqrt{d^2 - \left(g - \frac{1}{4} \sigma^2 - m - 2\rho\right)}\right)},\ 
        K \right\}.
    \end{align*}
\end{enumerate}

\newpage
\section{Additional Results for Numerical Experiments}
\label{appx: additional numerics}
\subsection{Fixed Stopping Levels}

\begin{table}[h] 
\centering
\begin{tabular}{ |c|c|c|c|c|c|c| } 
 \hline
 \rule{0pt}{10pt} With Fees & Fee Limit $K$ & 2 & 3 & 4 & 5 & 6 \\ 
 \hline 
 \multirow{3}{*}{\shortstack[l]{$r=0.14$, \\ $g=0$}} 
 & Scaled Threshold $c^\star$ & -18.931 & -24.516 & -28.462 & -31.519 & -34.016 \\ 
 & Optimal Exit Price $L^\star$ & 4.43e-8 & 3.00e-10 & 8.79e-12 & 5.71e-13 & 6.12e-14 \\
 & Optimal Value $V^\star$ &2.000 & 3.000 & 4.000 & 5.000 & 6.000 \\
 \hline
 \multirow{3}{*}{\shortstack[l]{$r=0$, \\ $g=0.13$}} 
 & Scaled Threshold $c^\star$ & -9.733 & -14.620 & -18.334 & -21.258 & -23.657 \\ 
 & Optimal Exit Price $L^\star$ & 1.66e-4 & 2.09e-6 & 7.56e-8 & 5.53e-9 & 6.47e-10\\
 & Optimal Value $V^\star$ &2.002 & 3.001 & 4.000 & 5.000 & 6.000 \\
 \hline
\end{tabular}

\medskip
\caption{This table presents a parametric analysis of the optimal exit scaled threshold $c^*$ and optimal exit price $L^*$ with transaction fees, under varying fee limit $K$, fixed $\rho=0.03$, $\sigma = \sqrt{0.8}$, $m=0.08$.}
\label{table: exit-threshold-vs-K}
\end{table}

\quad Table ~\ref{table: exit-threshold-vs-K} examines the impact of the transaction fee cap $K$, as defined in Proposition~\ref{fee structure}, on the investor’s optimal exit strategy under fixed reward rates and discount rates. We consider both rebasing $(r>0, g=0)$ and reward-bearing $(r=0, g>0)$ scenarios.

\quad Across all settings, we observe that the investor’s optimal objective value $V^\star$ closely matches the fee cap $K$, while the optimal exit price $L^\star= e^{c^\star\sigma}$ becomes vanishingly small. This implies that the effects of impermanent loss and opportunity cost on the overall objective are negligible. As a result, the investor chooses to stay in the AMM pool until the transaction fees have fully accumulated up to the cap $K$.

\quad This finding reinforces the results in Section~\ref{sc: implications for fee structure}, where we pointed out that fee accumulation is both rapid and substantial, and can dominate the investor’s overall payoff. It also further validates the stop-loss optimal exit strategy discussed in Section~\ref{sc: characterization of the optimal exit threshold}: the investor does not exit until the cumulative gain reaches an upper threshold, despite any intermediate loss.

\begin{table}[h] 
\centering
\begin{tabular}{ |c|c|c|c|c|c|c| } 
 \hline
 \rule{0pt}{10pt} & Exit Discount $m$ & 0.06 & 0.07 & 0.08 & 0.09 & 0.10 \\ 
 \hline 
 \multirow{3}{*}{No Fees} 
  & Scaled Threshold $c^\star$ & -1.340 & -1.319 & -1.298 & -1.278 & -1.259 \\ 
 & Optimal Exit Price $L^\star$ &0.302 &0.307 &0.313 &0.319 & 0.324 \\
 & Optimal Value $V^\star$ &0.169 &0.173 & 0.176 &0.179 & 0.182 \\
 \hline
 \multirow{3}{*}{With Fees} 
 & Scaled Threshold $c^\star$ & -10.475 & -13.662 & -18.931 & -29.355 & -63.688 \\ 
 &Optimal Exit Price $L^\star$ &8.53e-5 & 4.93e-6 & 4.43e-8 & 3.96e-12 & 1.82e-25 \\
 &Optimal Value $V^\star$ &2.002 &2.000 &2.000 &2.000 & 2.000 \\
 \hline
\end{tabular}

\medskip
\caption{This table presents a parametric analysis of the optimal exit scaled threshold $c^*$ and optimal exit price $L^*$ with and without transaction fees, under varying exit discount $m$, fixed $\rho=0.03$, $\sigma = \sqrt{0.8}$, $r=0.14$, $g=0$, and $K=2$.}
\label{table: exit-threshold-vs-m1}
\end{table}

\begin{table}[h] 
\centering
\begin{tabular}{ |c|c|c|c|c|c|c| } 
 \hline
 \rule{0pt}{10pt} & Exit Discount $m$ & 0.070 & 0.075 & 0.080 & 0.085 & 0.090 \\ 
 \hline 
 \multirow{3}{*}{No Fees} 
  & Scaled Threshold $c^\star$ & -1.180 & -1.170 & -1.160 & -1.151 & -1.142 \\ 
 & Optimal Exit Price $L^\star$ &0.348 &0.351 &0.354 &0.357 & 0.360 \\
 & Optimal Value $V^\star$ & 0.221 & 0.221 & 0.220 & 0.220 & 0.220 \\
 \hline
 \multirow{3}{*}{With Fees} 
 & Scaled Threshold $c^\star$ & -6.360 & -7.741 & -9.733 & -13.030 & -19.794 \\ 
 & Optimal Exit Price $L^\star$ & 3.39e-3 & 9.84e-4 & 1.66e-4 & 8.68e-6 & 2.05e-8\\
 & Optimal Value $V^\star$ &2.019 &2.008 &2.002 &2.000 & 2.000 \\
 \hline
\end{tabular}

\medskip
\caption{This table presents a parametric analysis of the optimal exit scaled threshold $c^*$ and optimal exit price $L^*$ with and without transaction fees, under varying exit discount $m$, fixed $\rho=0.03$, $\sigma = \sqrt{0.8}$, $r=0$, $g=0.13$, and $K=2$.}
\label{table: exit-threshold-vs-m2}
\end{table}

\begin{table}[h] 
\centering
\begin{tabular}{ |c|c|c|c|c|c|c| } 
 \hline
 \rule{0pt}{10pt} & Discount Rate $\rho$ & 0.010 & 0.015 & 0.020 & 0.025 & 0.030 \\ 
 \hline 
 \multirow{3}{*}{No Fees} 
  & Scaled Threshold $c^\star$ & -1.363 & -1.346 & -1.329 & -1.313 & -1.300 \\ 
 & Optimal Exit Price $L^\star$ & 0.295 & 0.300 & 0.305 & 0.309 & 0.313 \\
 & Optimal Value $V^\star$ & 0.180 &0.179 & 0.178 & 0.177 & 0.176 \\
 \hline
 \multirow{3}{*}{With Fees} 
 & Scaled Threshold $c^\star$ & -10.509 & -11.915 & -13.672 & -15.930 & -18.931 \\ 
 &Optimal Exit Price $L^\star$ & 8.27e-5 & 2.35e-5 & 4.89e-6 & 6.49e-7 & 4.43e-8 \\
 &Optimal Value $V^\star$ &2.003 &2.001 &2.000 &2.000 & 2.000 \\
 \hline
\end{tabular}

\medskip
\caption{This table presents a parametric analysis of the optimal exit scaled threshold $c^*$ and optimal exit price $L^*$ with and without transaction fees, under varying discount rate $\rho$, fixed $\sigma = \sqrt{0.8}$, $r=0.14$, $g=0$, $K=2$ and $m=0.08$.}
\label{table: exit-threshold-vs-rho1}
\end{table}

\begin{table}[h] 
\centering
\begin{tabular}{ |c|c|c|c|c|c|c| } 
 \hline
 \rule{0pt}{10pt} & Discount Rate $\rho$ & 0.010 & 0.015 & 0.020 & 0.025 & 0.030 \\ 
 \hline 
 \multirow{3}{*}{No Fees} 
  & Scaled Threshold $c^\star$ & -1.220 & -1.204 & -1.189 & -1.174 & -1.160 \\ 
 & Optimal Exit Price $L^\star$ & 0.336 & 0.341 & 0.345 & 0.350 & 0.354 \\
 & Optimal Value $V^\star$ & 0.234 & 0.230 & 0.227 & 0.224 & 0.220 \\
 \hline
 \multirow{3}{*}{With Fees} 
 & Scaled Threshold $c^\star$ & -4.446 & -5.349 & -6.400 & -7.763 & -9.733 \\ 
 & Optimal Exit Price $L^\star$ & 1.87e-2 & 8.34e-3 & 3.27e-3 & 9.65e-4 & 1.66e-4\\
 & Optimal Value $V^\star$ &2.068 & 2.039 & 2.020 & 2.008 & 2.002 \\
 \hline
\end{tabular}

\medskip
\caption{This table presents a parametric analysis of the optimal exit scaled threshold $c^*$ and optimal exit price $L^*$ with and without transaction fees, under varying discount rate $\rho$, fixed $\sigma = \sqrt{0.8}$, $r=0$, $g=0.13$, $K=2$ and $m=0.08$.}
\label{table: exit-threshold-vs-rho2}
\end{table}

\quad Tables~\ref{table: exit-threshold-vs-m1}, \ref{table: exit-threshold-vs-m2}, \ref{table: exit-threshold-vs-rho1}, and \ref{table: exit-threshold-vs-rho2} complement the results presented in Section~\ref{sc: parametric analysis of optimal exit}, specifically those in Tables~\ref{table: exit-threshold-vs-r} and \ref{table: exit-threshold-vs-g}. Together, these tables systematically examine the effect of the term $r + g - \rho - m$ in Assumption ~\ref{assump: stake via LSP}, which means the difference between the reward rates and the discount rates, on the investor’s optimal exit strategy. 

\quad Our parametric analysis shows a consistent pattern: the smaller the value of the difference $r + g - \rho - m$, the longer the investor tends to stay in the AMM pool. Tables~\ref{table: exit-threshold-vs-r} and \ref{table: exit-threshold-vs-g} explore how variations in reward rates 
$r$ and $g$ affect this difference. Tables~\ref{table: exit-threshold-vs-m1} and \ref{table: exit-threshold-vs-m2} further examine the role of the exit discount parameter $m$, while Tables~\ref{table: exit-threshold-vs-rho1} and \ref{table: exit-threshold-vs-rho2} investigate the impact of the discount rate $\rho$.

\subsection{Free Boundary Problem}
\label{appx sub: free boundary}

\quad As shown in Figures~\ref{fig:payoff decomp_r}–\ref{fig:total_comparison}, the total payoff (i.e., the investor’s portfolio value), represented by the second term in Eq.~\eqref{eq: Z(t,x)}, decreases more sharply for rebasing tokens with price $x$ as time $t$ increases. The difference in impermanent loss across different values of $t$ is relatively mild. The difference in their total payoffs mainly arises from the opportunity gain, which declines faster for rebasing tokens due to the coefficient $e^{(r - m - \rho)t}$, which increases with $t$ for rebasing tokens. For reward-bearing tokens, the corresponding coefficient becomes $e^{(-m - \rho)t}$, which decreases with $t$ and thus partially offsets the decline caused by rising prices. The difference between their rates of decrease becomes more pronounced as time $t$ increases. Hence, reward-bearing tokens can bear with higher prices, while rebasing tokens tend to take lower prices.


\begin{figure}[h]
    \centering
    \begin{subfigure}[t]{0.32\linewidth}
        \renewcommand\captionlabelfont{}%
        \centering
        \includegraphics[width=\linewidth]{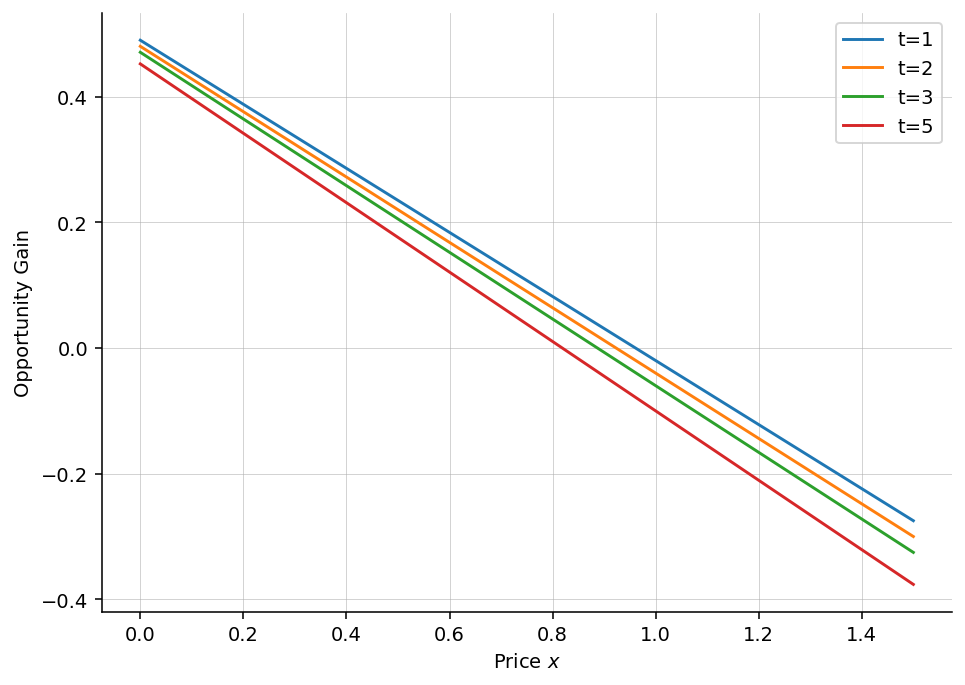}
        \caption{Opportunity Gain}
        \label{fig:OG_r}
    \end{subfigure}
    \begin{subfigure}[t]{0.32\linewidth}
        \renewcommand\captionlabelfont{}%
        \centering
        \includegraphics[width=\linewidth]{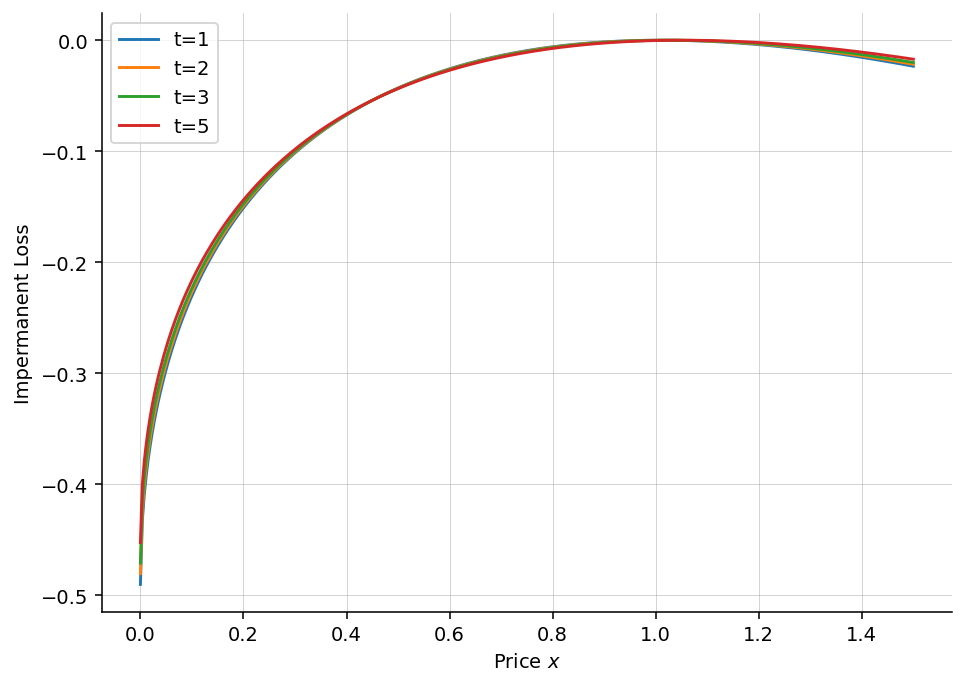}
        \caption{Impermanent Loss}
        \label{fig:IL_r}
    \end{subfigure}
    \begin{subfigure}[t]{0.32\linewidth}
        \renewcommand\captionlabelfont{}%
        \centering
        \includegraphics[width=\linewidth]{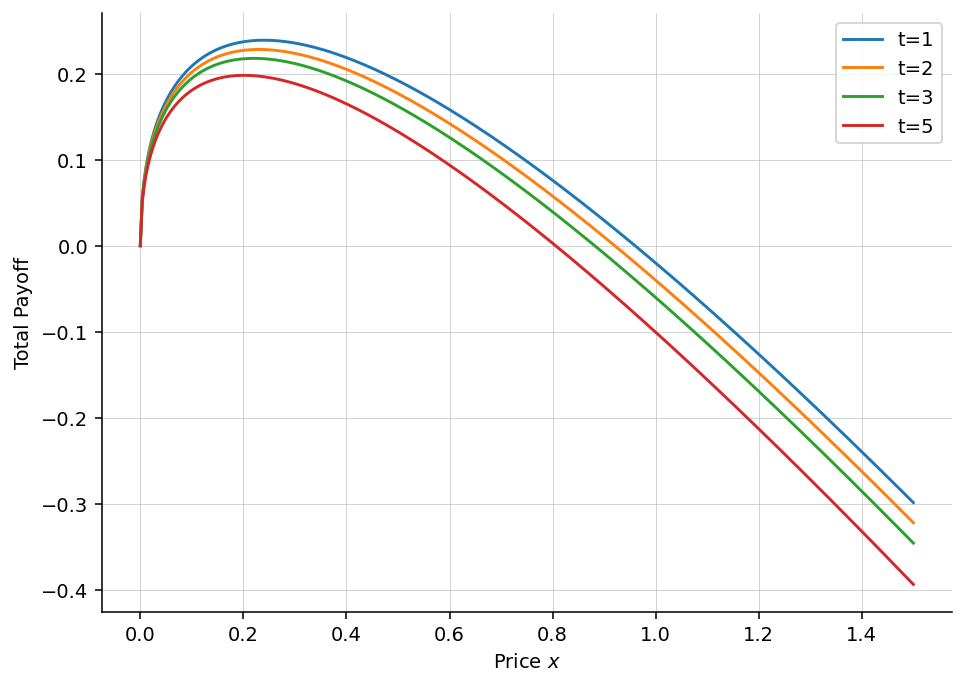}
        \caption{Total Payoff}
        \label{fig:total payoff_r}
    \end{subfigure}
    \caption{Decomposition of total payoff under rebasing tokens.}
    \label{fig:payoff decomp_r}
\end{figure}

\begin{figure}[h]
    \centering
    \begin{subfigure}[t]{0.32\linewidth}
        \renewcommand\captionlabelfont{}%
        \centering
        \includegraphics[width=\linewidth]{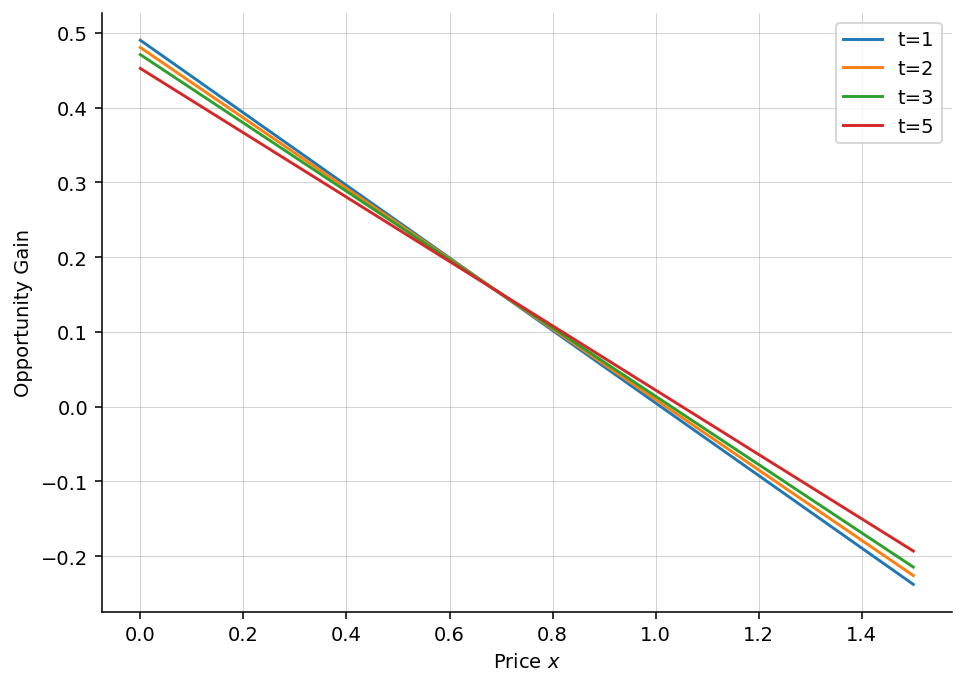}
        \caption{Opportunity Gain}
        \label{fig:OG_g}
    \end{subfigure}
    \begin{subfigure}[t]{0.32\linewidth}
        \renewcommand\captionlabelfont{}%
        \centering
        \includegraphics[width=\linewidth]{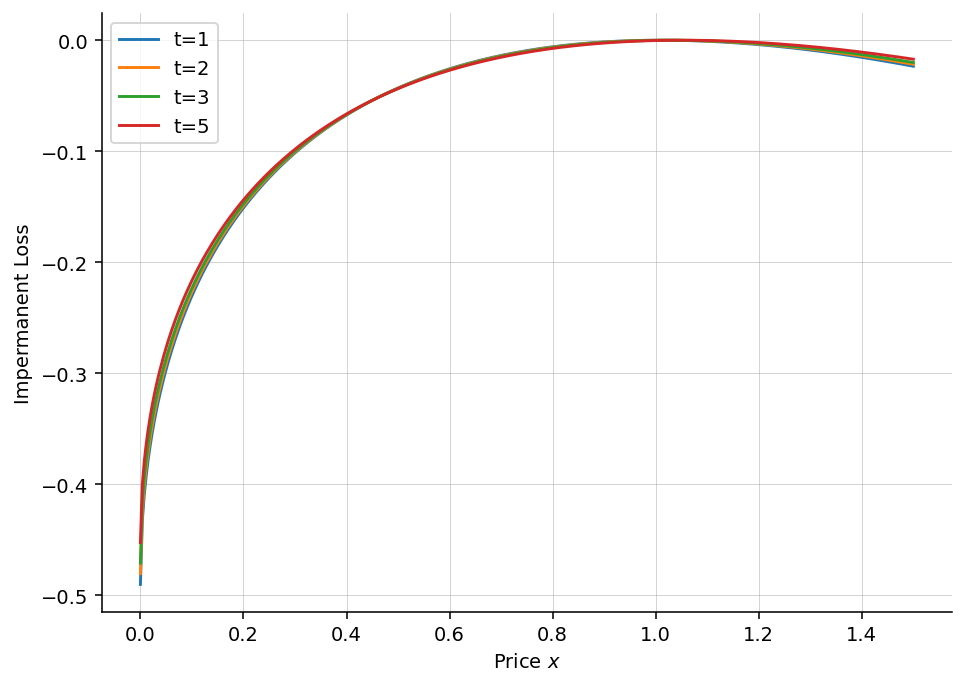}
        \caption{Impermanent Loss}
        \label{fig:IL_g}
    \end{subfigure}
    \begin{subfigure}[t]{0.32\linewidth}
        \renewcommand\captionlabelfont{}%
        \centering
        \includegraphics[width=\linewidth]{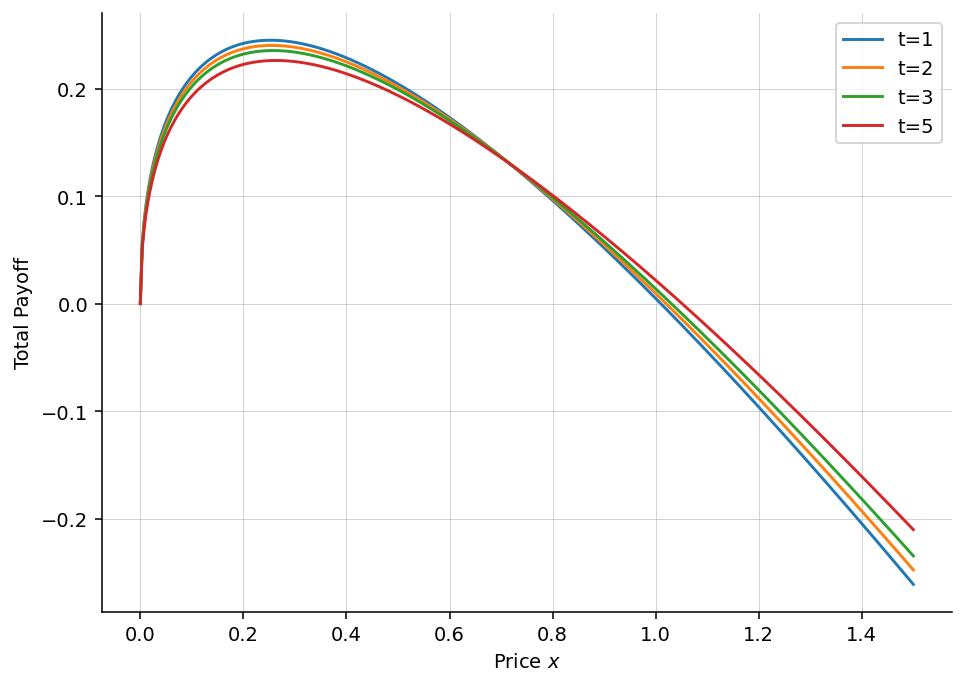}
        \caption{Total Payoff}
        \label{fig:total payoff_g}
    \end{subfigure}
    \caption{Decomposition of total payoff under reward-bearing tokens.}
    \label{fig:payoff decomp_g}
\end{figure}

\begin{figure}[h]
    \centering
    \begin{subfigure}[t]{0.32\linewidth}
        \renewcommand\captionlabelfont{}%
        \centering
        \includegraphics[width=\linewidth]{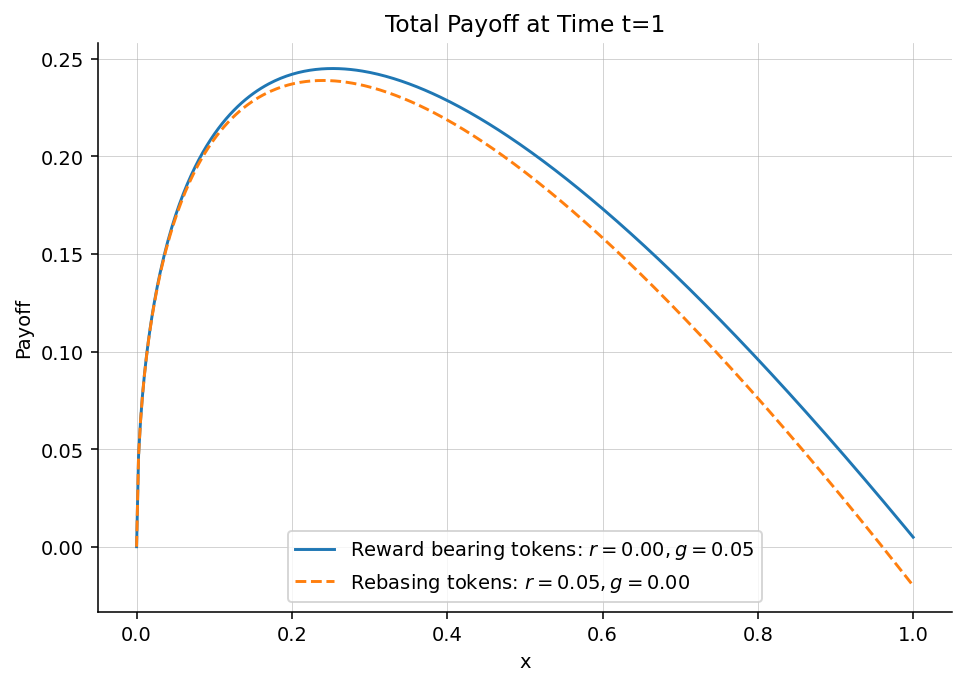}
        \caption{Time $t = 1$}
        \label{fig:total_1}
    \end{subfigure}
    \begin{subfigure}[t]{0.32\linewidth}
        \renewcommand\captionlabelfont{}%
        \centering
        \includegraphics[width=\linewidth]{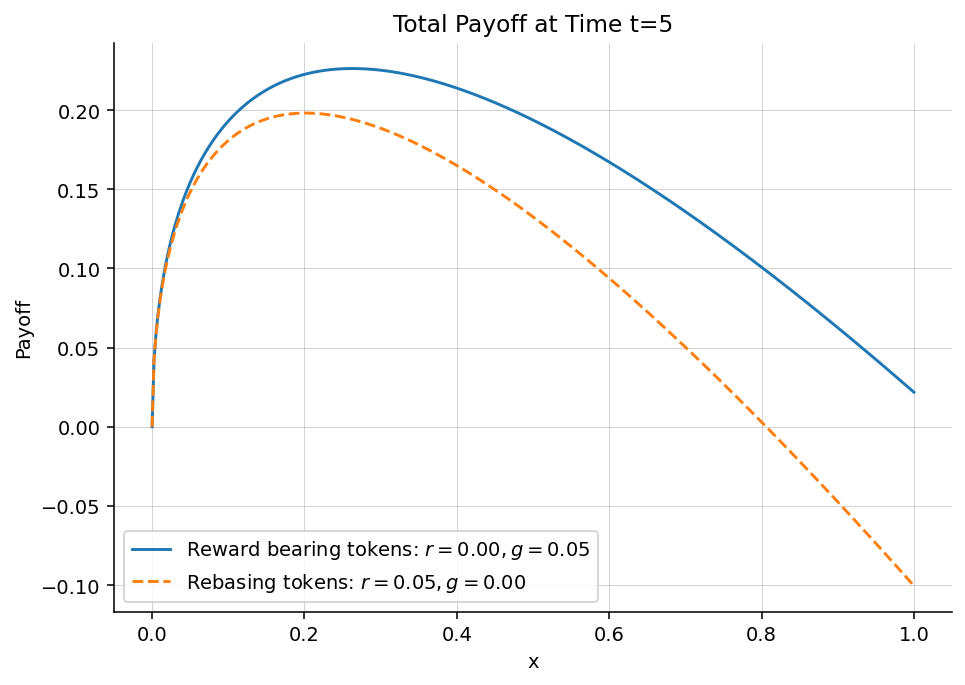}
        \caption{Time $t = 5$}
        \label{fig:total_2}
    \end{subfigure}
    \begin{subfigure}[t]{0.32\linewidth}
        \renewcommand\captionlabelfont{}%
        \centering
        \includegraphics[width=\linewidth]{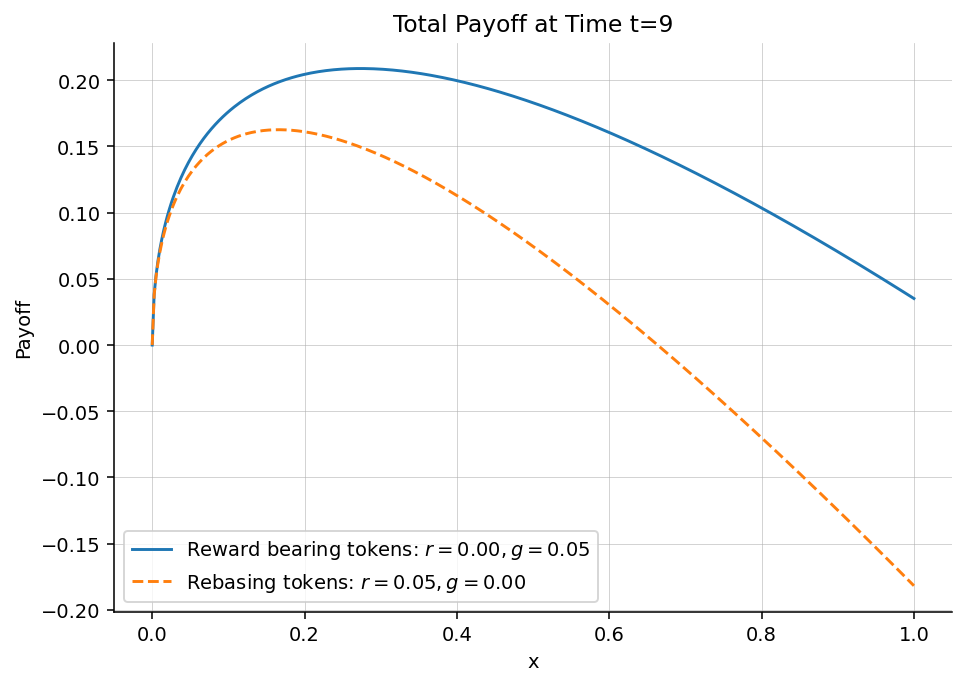}
        \caption{Time $t = 9$}
        \label{fig:total_3}
    \end{subfigure}
    \caption{Comparison of total payoffs for reward-bearing and rebasing tokens at time $t = 1, 5, 9$.}
    \label{fig:total_comparison}
\end{figure}

\end{document}